\documentclass[12pt]{article}
\usepackage{amsfonts,amsmath,amsxtra,amsthm}
\usepackage{amssymb}
\usepackage{bm}
\usepackage[utf8]{inputenc}
\usepackage{lscape}
\usepackage[normalem]{ulem}
\usepackage[russian,english]{babel}
\usepackage{xcolor}
\usepackage{pst-circ}
\usepackage{tikz, pgfplots}
\usepackage{tcolorbox}
\usepackage{mathrsfs}
\usepackage{cancel}

\usepackage{thmtools,thm-restate}

\renewcommand{\epsilon}{\varepsilon}

\def\C{\mathbb{C}}

\def\bg{\boldsymbol{\gamma}}

\def\beq{\begin{equation}}
\def\eeq{\end{equation}}
\def\beqq{\begin{equation*}}
\def\eeqq{\end{equation*}}

\def\bs{\begin{split}}
	\def\es{\end{split}}
\def\bg{\bm{\gamma}}
\def\bl{{\boldsymbol{\lambda}}}
\def\bbg{\underline{\bm{\gamma}}}

\def\bbl{\underline{\boldsymbol{\lambda}}}

\def\bo{\boldsymbol{\omega}}

\def\bx{\boldsymbol{x}}
\def\by{\boldsymbol{y}}
\def\bz{\boldsymbol{z}}

\def\bbx{\underline{\boldsymbol{x}}}
\def\bby{\underline{\boldsymbol{y}}}

\def\bbz{\underline{\boldsymbol{z}}}

\def\const{{2\pi\imath}}

\def\cbk{\color{black}}

\def\cbv{\color{violet}}

\def\eps{\epsilon}
\def\Im{\operatorname{Im}}
\def\g{\gamma}
\def\gg{{\hat{g}^\ast}}

\def\i{\hat{i}}

\def\K{{K}}
\def\KK{\hat{\K}}
\def\l{\lambda}

\def\mm{\hat{\mu}}

\def\o{\omega}

\def\R{\mathbb{R}}
\def\Re{\mathrm{Re}\,}
\def\Res{\operatorname{Res}}

\def\ve{\varepsilon}
\def\vf{\varphi}

\newtheorem{lemma}{Lemma}

\newtheorem*{proposition*}{Proposition}
\newtheorem{corollary}{Corollary}
\newtheorem*{theorem*}{Theorem}

\newcommand{\rf}[1]{(\ref{#1})}

\newcommand{\LL}[1]{\Lambda_{#1}}
\def\LLL{\hat{\Lambda}}

\def\QQ{\hat{Q}}

\begin{document}
\begin{center}
{\bf \large Baxter operators in Ruijsenaars hyperbolic system III. \\[4pt]
Orthogonality and completeness  of wave functions}

\bigskip

{\bf N. Belousov$^{\dagger\times}$, S. Derkachov$^{\dagger\times}$, S. Kharchev$^{\bullet\ast}$, S. Khoroshkin$^{\circ\ast}$
}\medskip\\
$^\dagger${\it Steklov Mathematical Institute, Fontanka 27, St. Petersburg, 191023, Russia;}\smallskip\\
$^\times${\it National Research University Higher School of Economics, Soyuza Pechatnikov 16, \\St. Petersburg, 190121, Russia;}\smallskip\\
$^\bullet${\it National Research Center ``Kurchatov Institute'', 123182, Moscow, Russia;}\smallskip\\
$^\circ${\it National Research University Higher School of Economics, Myasnitskaya 20, \\Moscow, 101000, Russia;}\smallskip\\
$^\ast${\it Institute for Information Transmission Problems RAS (Kharkevich Institute), \\Bolshoy Karetny per. 19, Moscow, 127994, Russia}

\end{center}

\begin{abstract}
\noindent 
In the previous paper we showed that the wave functions of
the quantum Ruijsenaars hyperbolic system diagonalize
Baxter Q-operators.
Using this property and duality relation we prove orthogonality and
completeness relations for the wave functions or, equivalently, unitarity
of the corresponding integral transform.
\end{abstract}

\tableofcontents

\section{Introduction}

\subsection{Ruijsenaars system: known results}
The Hamiltonians of the hyperbolic Ruijsenaars system \cite{R4} are given  by commuting symmetric difference operators
\begin{equation}
	\label{I2}
	H_r(\bx_n;g|\bo) = \sum_{\substack{I\subset[n] \\ |I|=r}}
	\prod_{\substack{i\in I \\ j\notin I}}
	\frac{\sh^{\frac{1}{2}}\frac{\pi}{\o_2}\left(x_i-x_j-\imath g\right)}
	{\sh^{\frac{1}{2}}\frac{\pi}{\o_2}\left(x_i-x_j\right)}
	\cdot T^{-\imath\o_1}_{I,x}\cdot \prod_{\substack{i\in I \\ j\notin I}}
	\frac{\sh^{\frac{1}{2}}\frac{\pi}{\o_2}\left(x_i-x_j+\imath g\right)}
	{\sh^{\frac{1}{2}}\frac{\pi}{\o_2}\left(x_i-x_j\right)}.
\end{equation}
They act on functions of $n$ coordinates
\begin{equation}
\bx_n = (x_1, \dots, x_n)
\end{equation}
and contain products of shift operators
\begin{equation}
	T^{a}_{I,x}=\prod_{i \in I} T_{x_i}^a, \qquad \left( T^{a}_{x_i} \, f \right)(x_1,\ldots,x_i,\ldots,x_n) = f(x_1,\ldots,x_i+a,\ldots,x_n)
\end{equation}
defined for any subset
\begin{equation}
I \subset [n] = \{1, \dots, n\}.
\end{equation}
One can also consider gauge equivalent Macdonald operators
\begin{equation}
	\label{I2a}
	M_r(\bx_n;g|\bo) = \sum_{\substack{I\subset[n] \\ |I|=r}}
	\prod_{\substack{i\in I \\ j\notin I}}
	\frac{\sh\frac{\pi}{\o_2}\left(x_i-x_j-\imath g\right)}
	{\sh\frac{\pi}{\o_2}\left(x_i-x_j\right)}
	\cdot T^{-\imath\o_1}_{I,x}
\end{equation}
collected into a generating function
\begin{align}
	M_n(\lambda;g) = \sum_{r=0}^{n} \lambda^{n-r}\,(-1)^r\,M_r(\bx_n;g|\bo).
\end{align}
Here we assume $M_0 = 1$.

Both families of operators are parametrized by three constants: periods $\bm{\omega}=(\omega_1, \omega_2)$ and coupling constant $g$,
which originally are supposed to be real positive, but we keep them to be complex unless otherwise stated and assume the following restrictions:
\beq\label{I0a} \Re \o_1 > 0, \qquad \Re \o_2 > 0,\qquad 0< \Re g<\Re \o_1+\Re \o_2  \eeq
and
\beq\label{I0b} \nu_g=\Re\frac{ g}{\o_1\o_2}>0.\eeq
 The gauge equivalence between Rujisenaars and  Macdonald operators  is established by means of the measure function
\beq\label{I5}
\mu(\bx_n)=\frac{1}{n!} \prod_{\substack{i,j=1 \\ i\not=j}}^n\mu(x_i-x_j),\qquad
\mu(x)=S_2(\imath x|\bo) S_2^{-1}(\imath x+g|\bo). \eeq
Here $S_2(z|\bo)$ is the double sine function, see Appendix \ref{S2}.
Namely,
\beq\label{I4}
\sqrt{\mu (\bx_n)} \,
M_r(\bx_n;g|\bo) \, \frac{1}{\sqrt{\mu  (\bx_n)}}=
H_r(\bx_n;g|\bo).
\eeq

Macdonald operators are symmetric with respect
to the bilinear pairing
\begin{align}\label{Mpair}
\bigl(\Phi|\Psi\bigr)_\mu =
\int_{\R^{n}} d\bx_{n}\,
\mu(\bx_{n})\,\Phi(-\bm{x}_n)
\Psi(\bm{x}_{n})
\end{align}
related to the measure $\mu(\bx_n)d\bx_n$, and with respect to the hermitian pairing
\begin{align}\label{Mpair2}
	\langle \Phi|\Psi\rangle_\mu =
	\int_{\R^{n}} d\bx_{n}\,
	\mu(\bx_{n})\,\overline{\Phi(\bm{x}_n)}
	\Psi(\bm{x}_{n})
\end{align}
as well once the periods $\o_1, \o_2$ and the coupling constant $g$ are supposed to be real.

There exists another family of commuting operators $Q_{n}(\lambda)$ parameterized
by $\lambda \in \mathbb{C}$ and called Baxter $Q$-operators. These are integral operators
\begin{equation}\label{I14}
\left[ Q_{n}(\lambda) f\right] (\bm{x}_n) =
\int_{\mathbb{R}^n} d\bm{y}_n \, Q(\bm{x}_n, \bm{y}_n; \lambda)\, f(\bm{y}_n)
\end{equation}
with the kernel
\beq\label{Qker} Q(\bx_n,\by_{n};\l) = d_n(g|\bo) \,e^{\const \l(\bbx_n-\bby_n)}
\K(\bx_n,\by_{n})\,\mu  (\by_{n}).
\eeq
Here and in what follows we denote
\begin{equation}\label{I6a}
\K(\bx_n,\by_m)=\prod_{i=1}^n \prod_{j = 1}^m \K(x_i-y_j)
\end{equation}
where $\K(x)$ is the following function
\beq\label{I6} \K(x) = S_2^{-1}\Bigl(\imath x +\frac{g^\ast}{2}\Big|\bo\Bigr)S_2^{-1}\Bigl(-\imath x+\frac{g^\ast}{2}\Big|\bo\Bigr) \eeq
and
\beq\label{I3b} g^\ast=\o_1+\o_2-g.\eeq
We also use notation
\begin{equation}
\bbx_n = x_1 + \ldots + x_n.
\end{equation}
The normalization constant
\begin{equation}\label{dconst}
d_n(g|\bo) = \left[ \sqrt{\omega_1 \omega_2} S_2(g|\bo) \right]^{-n}
\end{equation}
is inserted for the convenience: with this normalization eigenvalue of the $Q$-operator has the simplest form. Note that this constant differs from the one used in our previous work \cite{BDKK2} by the factor $1/n!$, which we instead include in the measure function \eqref{I5}.

Under assumptions \rf{I0a} and \rf{I0b} in \cite{BDKK} we established the equalities
\begin{align}\label{QQcomm}
Q_{n}(\lambda) \, Q_{n}(\rho) & = Q_{n}(\rho) \, Q_{n}(\lambda)\,,
\\[6pt] \label{QMcomm}
Q_{n}(\lambda) \, M_n(\rho; g) & = M_n(\rho; g) \,
Q_{n}(\lambda)\,.
\end{align}
Commutativity relations  \rf{QQcomm} and  \rf{QMcomm} suggest that $Q$-operators and Macdonald operators should have common eigenfunctions. They
can be constructed recursively
\begin{equation}\label{Psi-rec}
\Psi_{\bm{\lambda}_n}(\bm{x}_n) = \Lambda_{n}(\lambda_n) \, \Psi_{\bm{\lambda}_{n - 1}}(\bm{x}_{n - 1})\,,
\end{equation}
step by step, starting from the plane wave
$\Psi_{\lambda_1}(x_1) = e^{{2\pi \imath}{} \lambda_1 x_1}$,
by application of the raising operators $\Lambda_{n}(\lambda)$
\beq \label{I12}\Psi_{\bl_n}(\bx_n)=
\LL{n}(\l_n) \, \LL{n-1}(\l_{n-1})\cdots \LL{2}(\l_2) \, e^{\const \l_1x_1}.
\eeq
In some cases, to avoid misunderstanding, we use more detailed notation
for eigenfunctions $\Psi_{\bl_n}(\bx_n) = \Psi_{\bl_n}(\bx_n; g|\bo)$.
Raising operator $\Lambda_{n}(\lambda)$ is an integral operator
\beq\label{I11}\begin{split} \left[\LL{n}(\l)f\right](\bx_n) = \int_{\R^{n-1}}d\by_{n-1} \, \Lambda(\bx_n,\by_{n-1};\l) f(\by_{n-1})
\end{split}\eeq
with the kernel
\beq\label{Lker} \Lambda(\bx_n,\by_{n-1};\l)= d_{n - 1}(g|\bo) \,
e^{\const \l(\bbx_n-\bby_{n-1})}
\K(\bx_n,\by_{n-1})\,\mu  (\by_{n-1}).
\eeq
The normalization constant is the same as in \eqref{dconst}.
The kernel of the raising operator  is obtained from the corresponding integral kernel
of the $Q$-operator \eqref{Qker} by the  replacement $\by_{n} \to \by_{n-1}$
 and $d_{n}(g|\bo) \to d_{n - 1}(g|\bo)$.

Note that the operator $\LL{1}(\l_1)$ can be defined in a natural way
as the operator of multiplication by the function $e^{\const \l_1x_1}$ so that
the formula \eqref{I12} can be rewritten as
\beq \Psi_{\bl_n}(\bx_n)=
\LL{n}(\l_n) \, \LL{n-1}(\l_{n-1})\cdots \LL{2}(\l_2)\,\LL{1}(\l_1)\,1
\eeq
or in explicit terms
\begin{equation}\label{I12'}
\begin{aligned}
&\Psi_{\bl_n}(\bx_n; g|\bo)=  \prod_{k=1}^{n-1} \frac{d_{k}(g|\bo)}{k!}\,\int\limits_{\R^{\frac{n(n-1)}{2}}}\prod_{i=1}^{n-1} \Biggl\{
\prod_{1\leq l\neq m\leq i}\frac{S_2\big(\imath(x_{i,l}-x_{i,m})\big|\bo\big)}
{S_2\big(\imath(x_{i,l}-x_{i,m})+g\big|\bo\big)} \\[6pt]
& \times \prod\limits_{j=1}^{i+1}\prod\limits_{k=1}^{i}
\frac{S_2(\imath(x_{i+1,j}-x_{i,k})+\frac{\o_1+\o_2+g}{2}|\bo)}
{S_2(\imath(x_{i+1,j}-x_{i,k})+\frac{\o_1+\o_2-g}{2}|\bo)} \Biggr\} \,
e^{2\pi\imath \sum\limits_{i=1}^n \l_i \Bigl(\sum\limits_{j = 1}^{i} x_{i, j} - \sum\limits_{j = 1}^{i - 1}x_{i - 1, j}\Bigr)}
\prod_{i=1}^{n-1} \prod_{j = 1}^i dx_{i,j}
\end{aligned}
\end{equation}
Here the inversion formula \rf{inv} in used and we assume $x_{n, i} \equiv x_i$.

M. Hallnäs and S. Ruijsenaars proved that for real periods the function  \rf{I12} is an eigenfunction of Macdonald operators. In \cite{BDKK2} it is shown that under conditions \rf{I0a} and \rf{I0b} the function \rf{I12} is a joint eigenfunction of Macdonald operators and $Q$-operators as well
\begin{align}\label{I13}
& M_n(\lambda;g)\,\Psi_{\bl_n}(\bx_n) =
\prod_{j = 1}^n \left(\lambda-e^{2\pi\lambda_j \o_1}\right)\,
\Psi_{\bl_n}(\bx_n)\,,\\
\label{QPsi}
& Q_{n}(\lambda) \, \Psi_{ \bm{\lambda}_n }(\bm{x}_n) =
\prod_{j = 1}^n \KK(\lambda-\lambda_j)\, \Psi_{ \bm{\lambda}_n }(\bm{x}_n)\,.
\end{align}
\cbv The  \cbk eigenvalue of the operator $M_n(\lambda;g)$ coincides with the generating
function of elementary symmetric polynomials of arguments $e^{2\pi\lambda_j \o_1}$
and the eigenvalue of the $Q$-operator is expressed as the product of functions
\beq
\label{I6c} \KK(\l) = S_2^{-1}\Bigl(\imath \l +\frac{\hat{g}}{2}\Big|\hat{\bo}\Bigr)S_2^{-1}
\Bigl(-\imath \l+\frac{\hat{g}}{2}\Big|\hat{\bo}\Bigr)\,.
\eeq
Here we use the notation $\hat{a}=\frac{a}{\o_1\o_2}$ for any $a\in\C$
so that
\beq\label{hat}
\hat{\bo}=\left(\hat{\o}_1,\hat{\o}_2\right) = \left(\frac{1}{\o_2},\frac{1}{\o_1}\right) ,
\qquad \hat{g}=\frac{g}{\o_1\o_2},\qquad  \gg=\hat{\o}_1 + \hat{\o}_2 - \hat{g} = \frac{g^\ast}{\o_1\o_2}.
\eeq
The  proof of both relations \eqref{I13} and \eqref{QPsi} is based on
the intertwining relations
\begin{align}\label{MQLcomm}
&M_n(\lambda;g)\,\Lambda_n(\l_n)=(\lambda-e^{2\pi\lambda_n \omega_1})\,\Lambda_n(\l_n)\,M_{n-1}(\lambda;g)\,,\\[6pt]
&Q_{n}(\lambda) \, \Lambda_{n}(\lambda_n) = \KK(\lambda - \lambda_n)\,
\Lambda_{n}(\lambda_n) \, Q_{n-1}(\lambda)\,.
\end{align}

In \cite{BDKK2} we proved a number of symmetries of the wave function \rf{I12}. It is clearly a symmetric function of the variables $x_k$. Besides, it is symmetric function of spectral variables $\l_k$ due to exchange relation
\begin{equation*}
	\Lambda_{n}(\lambda_n) \, \Lambda_{n - 1}(\lambda_{n-1}) =
\Lambda_{n}(\lambda_{n-1}) \, \Lambda_{n - 1}(\lambda_n).
\end{equation*}
Moreover, it enjoys the remarkable duality property with respect to interchange of space and spectral variables. Similar to  \eqref{I6c} define the function
\beq\label{I5d}
 \mm(\l) = S_2(\imath \l|\hat{\bo}) S_2^{-1}(\imath \l+\gg|\hat{\bo})\,
\eeq
and corresponding products
\begin{equation}\label{I6d}
\mm(\bl_n)= \frac{1}{n!}\,\prod_{\substack{i,j = 1\\i\not=j}}^n \mm(\l_i-\l_j),\qquad	\KK(\bl_n,\bg_m)=\prod_{i=1}^n \prod_{j = 1}^m \KK(\l_i-\g_j).
\end{equation}
Next we introduce counterparts of the Macdonald operators and their generating function
that act on functions of spectral variables $\lambda_j$
\begin{align*}
&M_r(\bl_n;\hat{g}^*|\hat{\bo}) = \sum_{\substack{I\subset[n] \\ |I|=r}}
\prod_{\substack{i\in I \\ j\notin I}}
\frac{\sh\pi \o_1\left(\lambda_i-\lambda_j-\imath \hat{g}^*\right)}
{\sh \pi \o_1\left(\lambda_i-\lambda_j\right)}
\cdot T^{-\frac{\imath}{\o_2}}_{I,\lambda}\,,\\[6pt]
&\hat{M}_n(x;\hat{g}^*) = \sum_{r=0}^{n} x^{n-r}\,(-1)^r\,
M_r(\bl_n;\hat{g}^*|\hat{\bo})
\end{align*}
and counterparts of the $Q$-operator and raising $\Lambda$-operator as integral operators
\begin{equation}\label{I15d}\begin{aligned}
\bigl[\QQ_n(x) \varphi\bigr] (\bm{\l}_n)& = \int_{\mathbb{R}^n} d\bm{\gamma}_n \, \QQ(\bm{\l}_n, \bm{\gamma}_n; x) \,\varphi(\bm{\gamma}_n),\\[7pt]
	 \bigl[\LLL_{n}(x)\varphi\bigr](\bl_n)& = \int_{\R^{n-1}}d\bg_{n-1} \, \LLL(\bl_n,\bg_{n-1};x) \,\varphi(\bg_{n-1})
	\end{aligned}
\end{equation}
with the kernels
\beq\label{I14d}\begin{split} \QQ(\bl_n,\bg_{n};x)&=
d_n(\hat{g}^*|\hat{\bo})\,e^{\const x (\bbl_n-\bbg_n)}	
\,\KK(\bl_n,\bg_{n})\,\hat{\mu}(\bg_{n}),\\[10pt]
\LLL(\bl_n,\bg_{n-1};x)&= d_{n - 1}(\hat{g}^*|\hat{\bo})\,
e^{\const x (\bbl_n-\bbg_{n-1} )}\, \KK(\bl_n,\bg_{n-1})\, \hat{\mu}  (\bg_{n-1}).
\end{split}
\eeq
In addition to \eqref{I0a}, \eqref{I0b} assume
\begin{equation}
\nu_{g^*} = \Re \frac{g^*}{\o_1 \o_2} > 0.
\end{equation}
Altogether we have two pairs of inequalities on coupling constant
\begin{equation}
0 < \Re g < \Re \o_1 + \Re \o_2, \qquad 0 < \Re \hat{g}^* < \Re \hat{\o}_1 + \Re \hat{\o}_2.
\end{equation}
Notice that they are equivalent in the case of real periods $\o_1 ,\o_2$. Then define
\begin{equation}\label{Psi-rec2}
	\hat{\Psi}_{\bx_n}(\bl_n) := \Psi_{\bx_n}(\bl_n;\gg|\hat{\bo}) = \LLL_{n}(x_n) \, \hat{\Psi}_{\bx_{n - 1}}(\bl_{n - 1}), \qquad \hat{\Psi}_{x_1}(\l_1) = e^{{2\pi \imath}{} \lambda_1 x_1}
\end{equation}
or in explicit terms
\begin{equation*}
\begin{aligned}
&\hat{\Psi}_{\bx_n}(\bl_n)=  \prod_{k=1}^{n-1} \frac{d_{k}(\hat{g}^*|\hat{\bo})}{k!}\,\int\limits_{\R^{\frac{n(n-1)}{2}}}\prod_{i=1}^{n-1} \Biggl\{
\prod_{1\leq l\neq m\leq i}\frac{S_2\big(\imath(\l_{i,l}-\l_{i,m})\big|\hat{\bo}\big)}
{S_2\big(\imath(\l_{i,l}-\l_{i,m})+\hat{g}^*\big|\hat{\bo}\big)} \\[6pt]
& \times \prod\limits_{j=1}^{i+1}\prod\limits_{k=1}^{i}
\frac{S_2(\imath(\l_{i+1,j}-\l_{i,k})+\frac{\o_1+\o_2+g^*}{2\o_1 \o_2}|\hat{\bo})}
{S_2(\imath(\l_{i+1,j}-\l_{i,k})+\frac{\o_1+\o_2-g^*}{2 \o_1 \o_2}|\hat{\bo})} \Biggr\} \,
e^{2\pi\imath \sum\limits_{i=1}^n x_i \Bigl(\sum\limits_{j = 1}^{i} \l_{i, j} - \sum\limits_{j = 1}^{i - 1}\l_{i - 1, j}\Bigr)}
\prod_{i=1}^{n-1} \prod_{j = 1}^i d\l_{i,j}
\end{aligned}
\end{equation*}
where we assume $\l_{n,i} \equiv \l_i$. The following statement was conjectured in \cite{HR1} and proved in \cite{BDKK2}.
\begin{theorem*}\label{thm-duality}
	The wave function $\Psi_{\bm{\lambda}_n}(\bm{x}_n)$ satisfies duality relation
	\begin{equation}\label{duality}
		\Psi_{\bm{\lambda}_n}(\bm{x}_n; g|\bo) = \Psi_{\bm{x}_n}(\bm{\lambda}_n; \gg|\hat{\bo}).
	\end{equation}
\end{theorem*}

The statement of duality has several direct corollaries.
The wave function  \rf{I12} is an eigenfunction of both operators $Q_{n}(\lambda)$ and $\QQ_{n}(x)$
\begin{align}\label{QQPsi1}
	Q_{n}(\lambda) \, \Psi_{ \bm{\lambda}_n }(\bm{x}_n) &= \prod_{j = 1}^n \KK(\lambda-\lambda_j) \, \Psi_{ \bm{\lambda}_n }(\bm{x}_n),\\
	\label{QQPsi2}
	\QQ_{n}(x) \, \Psi_{ \bm{\lambda}_n }(\bm{x}_n) &= \prod_{j = 1}^n \K(x-x_j) \, \Psi_{ \bm{\lambda}_n }(\bm{x}_n).
\end{align}
as well as of both families of Macdonald operators $M(\lambda;g|\bo)$ and $\hat{M}(x;\hat{g}^*|\hat{\bo})$:
	\beq\label{MM}\begin{split}
	 M(\lambda;g|\bo)\,\Psi_{\bl_n}(\bx_n)&=
\prod_{j=1}^{n}\left(\lambda-e^{2\pi\lambda_j\o_1}\right)\,\Psi_{\bl_n}(\bx_n),\\[5pt]
	 \hat{M}(x;\hat{g}^*|\hat{\bo})\, \Psi_{\bl_n}(\bx_n)&=
\prod_{j=1}^{n}\Bigl(x-e^{\frac{2\pi x_j}{\o_2}}\,\Bigr)\,\Psi_{\bl_n}(\bx_n).
\end{split}
\eeq
In other words, it solves bispectral problems for the pair of dual $Q$-operators and for the pair of dual Macdonald operators.  In non-relativistic limit the bispectrality property of the wave function of Sutherland model was established in \cite{KK1,KK2}. The  duality \rf{MM} \cbk provides us with two integral representations \eqref{Psi-rec}, \eqref{Psi-rec2}, recursive with respect to $x_j$ and $\lambda_j$ variables correspondingly
\begin{equation}
\Psi_{\bm{\lambda}_n}(\bm{x}_n) = \Lambda_{n}(\lambda_n) \, \Psi_{\bm{\lambda}_{n - 1}}(\bm{x}_{n - 1}) = \LLL_{n}(x_n) \, \Psi_{\bm{\lambda}_{n - 1}}(\bm{x}_{n - 1}).
\end{equation}

\subsection{Orthogonality and completeness: expectations}

In the present paper we prove orthogonality and
completeness relations for the eigenfunctions
$\Psi_{\bm{\lambda}_n}(\bm{x}_n)$. To \cbk start with let us itemize what we should expect on the formal level. Assume  that periods $\o_1 ,\o_2$ and the coupling constant $g$ are real.

Macdonald operators are formally self-adjoint with respect to the scalar product \rf{Mpair2}
so that we expect the orthogonality relation for the eigenfunctions in the following general form
\begin{align}\label{ort0}
\int_{\R^{n}} d\bx_{n}\,
\mu(\bx_{n})\,
\overline{\Psi}_{\bm{\lambda}'_n}(\bm{x}_n)
\Psi_{\bm{\lambda}_{n}}(\bm{x}_{n}) =
A(\bm{\lambda}_{n})\,\delta (\bm{\lambda}'_{n}, \bm{\lambda}_{n})
\end{align}
where $A(\bm{\lambda}_{n})$ is normalization constant and
\begin{align}
\delta (\bm{\lambda}'_{n}, \bm{\lambda}_{n}) = \frac1{n!}\,
\sum_{w\in S_n} \prod_{k=1}^n \delta(\lambda'_k-\lambda_{w(k)})
\end{align}
is the kernel of the identity operator on the space
of the symmetric functions of spectral variables $\bm{\lambda}_{n} \in \R^n$. Due to the duality \eqref{duality} we expect analogous orthogonality relation over spectral variables
\beq \label{orth2} \int_{\R^{n}} d\bl_{n}\,
\hat{\mu}(\bl_{n})\,
\overline{\Psi}_{\bm{\lambda}_n}(\bm{x}_n')
\Psi_{\bm{\lambda}_{n}}(\bm{x}_{n}) =
B(\bm{x}_{n})\,\delta (\bm{x}_{n}', \bm{x}_{n}) \eeq
where
\begin{align}\label{dx}
	\delta (\bm{x}'_{n}, \bm{x}_{n}) = \frac1{n!}\,
	\sum_{w\in S_n} \prod_{k=1}^n \delta(x'_k-x_{w(k)})
\end{align}
is the kernel of the identity operator on the space
of the symmetric functions of coordinate variables $\bm{x}_{n} \in \R^n$.
The relation \rf{orth2} can be regarded as a completeness relation. By formal linear algebra manipulations we then see that
  \begin{align}\label{ort3}
  	A(\bm{\lambda}_{n}) = \frac{c_n(g |\bm{\omega})}{\hat{\mu}(\bm{\lambda}_{n})}
  	\ \ ;\ \
  	B(\bm{x}_{n}) = \frac{c_n(g |\bm{\omega})}{\mu(\bm{x}_{n})}
  \end{align}
where $c_n(g |\bm{\omega})$ is some constant which does not
depends on $\bm{\lambda}_{n}$ and $\bm{x}_{n}$.
One of the main results of this paper is the proof of the formulae \eqref{ort0}, \eqref{orth2}, which shows that actually $c_n(g|\bo) = 1$.

The two different recursive integral representations of the wave function
 \begin{align}\label{itx}
 	\Psi_{\bm{\lambda}_{n+1}}(\bm{x}_{n+1}) =&
 	\int_{\R^{n}}d\by_{n} \,
 	\Lambda(\bx_{n+1},\by_{n};\lambda_{n+1})\,\Psi_{\bm{\lambda}_{n}}(\bm{y}_{n}),\\[6pt]
 \label{itlam}
 	\Psi_{\bm{\lambda}_{n+1}}(\bm{x}_{n+1}) =&
 	\int_{\R^{n}}d\bm{\gamma}_{n} \,
 	\hat{\Lambda}(\bm{\lambda}_{n+1},\bm{\gamma}_{n};x_{n+1})\,
 	\Psi_{\bm{\gamma}_{n}}(\bm{x}_{n})
 \end{align}
together with orthogonality relations \eqref{ort0}, \eqref{orth2} allow us to express the kernel of raising operators $\Lambda_{n+1}(\lambda_{n+1})$ and  $\LLL_{n+1}(x_{n+1})$, as well as of the corresponding $Q$-operators via scalar products of wave functions. Indeed, multiplying~\eqref{itlam} by
$\overline{\Psi}_{\bm{\g}_{n}}(\bm{x}_{n})$ and integrating
over~$\bm{x}_{n}$ with the needed measure $\mu(\bx_{n})$ we obtain
\begin{align*}&
\int_{\R^{n}}d\bm{x}_{n}\,\mu(\bx_{n})\, \overline{\Psi}_{\bm{\g}_{n}}(\bm{x}_{n})\,
\Psi_{\bm{\lambda}_{n+1}}(\bm{x}_{n+1}) = \\[4pt]
&\int_{\R^{n}}d\bm{\nu}_{n} \,
\hat{\Lambda}(\bm{\lambda}_{n+1},\bm{\nu}_{n};x_{n+1})\,
\int_{\R^{n}}d\bm{x}_{n}\,\mu(\bx_{n})\, \overline{\Psi}_{\bm{\g}_{n}}(\bm{x}_{n})\,
\Psi_{\bm{\nu}_{n}}(\bm{x}_{n}) = \\[4pt]
&c_n(g|\bo) \int_{\R^{n}}d\bm{\nu}_{n} \,
\hat{\Lambda}(\bm{\lambda}_{n+1},\bm{\nu}_{n};x_{n+1})\,
\left(\hat{\mu}(\bm{\g}_{n})\right)^{-1}\,
\delta (\bm{\g}_{n}, \bm{\nu}_{n}) =\\[4pt]
&c_n(g|\bo)\left(\hat{\mu}(\bm{\g}_{n})\right)^{-1}\,
\hat{\Lambda}(\bm{\lambda}_{n+1},\bm{\g}_{n};x_{n+1})\,.
\end{align*}
Finally, the expression for the integral kernel in dual space is fixed
in the following form by the scalar product
\begin{align}\label{kdual}
\hat{\Lambda}(\bm{\lambda}_{n+1},\bm{\g}_{n};x_{n+1}) =
c_n^{-1}(g|\bo) \, \hat{\mu}(\bm{\g}_{n})\,
\int_{\R^{n}}d\bm{x}_{n}\,\mu(\bx_{n})\, \overline{\Psi}_{\bm{\g}_{n}}(\bm{x}_{n})\,
\Psi_{\bm{\lambda}_{n+1}}(\bm{x}_{n+1}).
\end{align}
In a similar way we obtain the expression for the integral kernel of $\QQ$-operator.
We start from the relation \eqref{QQPsi2}
\begin{align*}
\int_{\mathbb{R}^n} d\bm{\nu}_n \, \QQ(\bm{\l}_n, \bm{\nu}_n; x)\,\Psi_{\bm{\nu}_{n}}(\bm{x}_{n}) =
\K(x,\bx_{n}) \,\Psi_{\bm{\lambda}_{n}}(\bm{x}_{n})\,,
\end{align*}
where
\begin{equation}
\K(x,\bx_{n}) = \prod_{i=1}^n \K(x-x_{i}),
\end{equation}
multiply it by $\overline{\Psi}_{\bm{\g}_{n}}(\bm{x}_{n})$ and integrate
over $\bm{x}_{n}$ with the needed measure $\mu(\bx_{n})$
\begin{align*}
&\int_{\mathbb{R}^n} d\bm{x}_n \,\mu(\bx_{n})\,
\overline{\Psi}_{\bm{\g}_{n}}(\bm{x}_{n})
\K(x,\bx_{n}) \,\Psi_{\bm{\lambda}_{n}}(\bm{x}_{n}) = \\[4pt]
&\int_{\mathbb{R}^n} d\bm{\nu}_n \, \QQ(\bm{\l}_n, \bm{\nu}_n; x)\,
\int_{\mathbb{R}^n} d\bm{x}_n \,\mu(\bx_{n})\,\overline{\Psi}_{\bm{\g}_{n}}(\bm{x}_{n})\,
\Psi_{\bm{\nu}_{n}}(\bm{x}_{n}) = \\[4pt]
&c_n(g|\bo)\int_{\mathbb{R}^n} d\bm{\nu}_n \, \QQ(\bm{\l}_n, \bm{\nu}_n; x)\,\left(\hat{\mu}(\bm{\g}_{n})\right)^{-1}\,
\delta (\bm{\g}_{n}, \bm{\nu}_{n}) =\\[4pt]
&c_n(g|\bo) \left(\hat{\mu}(\bm{\g}_{n})\right)^{-1}\,
\QQ(\bm{\l}_n, \bm{\g}_n; x)\,,
\end{align*}
which gives
\begin{align}\label{hatQ}
\QQ(\bm{\l}_n, \bm{\g}_n; x) =
c_n^{-1}(g|\bo) \, \hat{\mu}(\bm{\g}_{n})\,
\int_{\mathbb{R}^n} d\bm{x}_n \,\mu(\bx_{n})\,
\overline{\Psi}_{\bm{\g}_{n}}(\bm{x}_{n})
\K(x,\bx_{n}) \,\Psi_{\bm{\lambda}_{n}}(\bm{x}_{n}).
\end{align}
In the next section we calculate the integrals on the right hand sides \eqref{kdual}, \eqref{hatQ} explicitly in a more general setting of complex periods and coupling constant
\begin{align}\label{I0}
&\int_{\R^{n}}d\bm{x}_{n}\,\mu(\bx_{n})\,
\Psi_{\bm{\gamma}_{n}}(-\bm{x}_{n})\,
\Psi_{\bm{\lambda}_{n+1}}(\bm{x}_{n+1}) = \frac{1}{d_n(g|\bo)}\,e^{2\pi\imath \left(\bbl_{n+1}-
\bbg_n \right)\,x_{n+1}}\,
\hat{K}(\bm{\lambda}_{n+1},\bm{\gamma}_n)\,,\\[8pt]
&\label{J0}\int_{\R^{n}}d\bm{x}_{n}\,\mu(\bx_{n})\,
\Psi_{\bm{\gamma}_{n}}(-\bm{x}_{n})\,
\K(x,\bx_{n})\,\Psi_{\bm{\lambda}_{n}}(\bm{x}_{n}) = \,\frac{1}{d_n(g|\bo)}\,e^{2\pi\imath\left(\bbl_{n}-
\bbg_n\right) x }\,
\hat{K}(\bm{\lambda}_{n},\bm{\gamma}_n)\,.
\end{align}
Note that assuming all parameters to be real
\begin{equation}
\Psi_{\bg_n}(-\bx_n) = \overline{\Psi}_{\bg_n}(\bx_n).
\end{equation}
Therefore the above formulas are consistent with an assumption  $c_n(g|\bo)=1$ and with the relations~\rf{I14d}.

The integral \rf{J0} plays a very important role in our work: it is used for the construction of a delta sequence converging to the scalar product \rf{ort0}.

\subsection{Statements} Our main result may be formulated as follows. Let the periods and the coupling constant be real satisfying the conditions \rf{I0a}. Denote by $L^2_{\mathrm{sym}}\left(\mathbb R^{n}, \mu\right)$ the Hilbert space of symmetric functions square integrable with respect to the measure $\mu(\bx_n)d\bx_n$ and by $L^2_{\mathrm{sym}}\left(\mathbb R^{n}, \hat{\mu}\right)$ the Hilbert space of symmetric functions square integrable with respect to  the measure $\hat{\mu}(\bl_n)d\bl_n$. The corresponding norms are
\begin{align}
\| \phi \|_{\mu}^2 = \int_{\R^{n}} d\bx_n \, \mu(\bx_n) \, | \phi (\bx_n) |^2, \\[4pt]
 \| \vf \|_{\hat{\mu}}^2 = \int_{\R^{n}} d\bl_n \, \hat{\mu}(\bl_n) \, | \vf (\bl_n) |^2.
\end{align}
For any finitely supported smooth functions $\vf(\bm{\lambda}_n)$ and $\phi(\bm{x}_n)$
define linear operators
\begin{align}
&(T \varphi) (\bm{x}_n) = \int d\bm{\lambda}_n \, \hat{\mu}(\bm{\lambda}_n) \, \Psi_{\bm{\lambda}_n} (\bm{x}_n) \, \varphi(\bm{\lambda}_n)\,, \\
&(S \phi) (\bm{\lambda}_n) = \int d\bm{x}_n \, \mu(\bm{x}_n) \,
\Psi_{\bm{\lambda}_n} (\bm{x}_n) \, \phi(\bm{x}_n)
\,.
\end{align}
Denote also by $R$ the total reflection operator in $\R^n$,
\beq \label{R} Rf(\bl_n)= f(-\bl_n), \qquad -\bl_n=(-\l_1,\ldots, -\l_n). \eeq

\begin{restatable}{thm}{Ttheorem}\label{T-theorem}
	For any smooth, compactly supported function $\varphi$ on $\mathbb{R}^{n}$,
	$T \varphi  \in L^2_{\mathrm{sym}}\left(\mathbb R^{n}, \mu\right)$ and the following relation holds
	\begin{align}\label{unit-T}
		\| T \varphi\|_{\mu}^2=\|{\varphi}\|_{\hat{\mu}}^2\,.
	\end{align}
As such, $T$ extends to a linear isometry $T \colon
	L^2_{\mathrm{sym}}\left(\mathbb R^{n}, \hat{\mu}\right)
	\mapsto L^2_{\mathrm{sym}}\left(\mathbb R^{n}, \mu\right)$.
\end{restatable}
\newpage

\begin{restatable}{thm}{Stheorem}\label{S-theorem} \-
\begin{enumerate}
	\item[(i)] For any smooth, compactly supported function $\phi$ on $\mathbb{R}^{n}$,
	$S \phi  \in L^2_{\mathrm{sym}}\left(\mathbb R^{n}, \hat{\mu}\right)$ and the following relation holds
	\begin{align}\label{unit-S}
	\| S \phi\|_{\hat{\mu}}^2=\|{\phi}\|_{\mu}^2\,.
	\end{align}
	As such, $S$ extends to a linear isometry $S \colon
	L^2_{\mathrm{sym}}\left(\mathbb R^{n}, \mu\right)
	\mapsto L^2_{\mathrm{sym}}\left(\mathbb R^{n}, \hat{\mu}\right)$.
	
	\item[(ii)] Operators $S$  and $T$ establish unitary isomotphisms between the spaces 	$L^2_{\mathrm{sym}}\left(\mathbb R^{n}, \mu\right)$ and $L^2_{\mathrm{sym}}\left(\mathbb R^{n}, \hat{\mu}\right)$ and are connected  by the relation
	\beq\label{dseq5} S=RT^\ast =RT^{-1}, \eeq
	so that in $L^2$-sense
	\beq \label{dseq6} STf(\bl_n)=f(-\bl_n),\qquad TS g(\bx_n)=g(-\bx_n). \eeq
\end{enumerate}
\end{restatable}
\cbk
In particular, we have the relations of orthogonality and completeness
\begin{align} \label{St1}&\int_{\R^{n}} d\bx_{n}\,
\mu(\bx_{n})\,
\overline{\Psi}_{\bm{\lambda}'_n}(\bm{x}_n)
\Psi_{\bm{\lambda}_{n}}(\bm{x}_{n}) =
\hat{\mu}^{-1}(\bm{\lambda}_{n})\,\delta (\bm{\lambda}'_{n}, \bm{\lambda}_{n}),\\[6pt]
 &\int_{\R^{n}} d\bl_{n}\,
\hat{\mu}(\bl_{n})\,
\overline{\Psi}_{\bm{\lambda}_n}(\bm{x}_n')
\Psi_{\bm{\lambda}_{n}}(\bm{x}_{n}) =
\mu^{-1}(\bm{x}_{n})\,\delta (\bm{x}_{n}', \bm{x}_{n}). \label{St2}
\end{align}
The proof of Theorem \ref{T-theorem} follows the strategy of the arguments from \cite{DKM}. Namely, the spectral properties of $Q$-operators provide us with two sets of explicitly computable integrals \rf{I} and \rf{J}, which are then used for the construction of  specific delta sequences converging to the scalar products between the wave functions. Note an important role of Lemma \ref{lemma:delta}, which gives us a multidimensional version of Sokhotski - Plemelj formula. It also appeared first in \cite{DKM}. The subtle point is to prove that this delta sequence indeed converges to a scalar product in a sense of distributions. Here, similar to \cite{DKM} we use Fatou's lemma.

The first statement of Theorem \ref{S-theorem} essentially follows from Theorem \ref{T-theorem} due to the remarkable duality symmetry of the eigenfunctions \eqref{duality}.

Unless the main theorems are proved for real parameters, most of intermediate statements are proved for complex parameters satisfying \rf{I0a} and \rf{I0b}. First of all, we have explicit calculations of certain integrals.
Set
\begin{align}\label{I}
	&I(\bm{\gamma}_{n}\,,\bm{\lambda}_{n+1}; x) = \int_{\R^{n}}d\bm{x}_{n}\,\mu(\bx_{n})\,
	\Psi_{\bm{\gamma}_{n}}(-\bm{x}_{n})\,
	\Psi_{\bm{\lambda}_{n+1}}(\bm{x}_{n}\,,x)\,,\\[6pt]
	\label{J}
	&J(\bm{\gamma}_{n}\,,\bm{\lambda}_{n}; x) = \int_{\R^{n}}d\bm{x}_{n}\,\mu(\bx_{n})\,
	\Psi_{\bm{\gamma}_{n}}(-\bm{x}_{n})\,
	\K(x,\bx_{n})\,\Psi_{\bm{\lambda}_{n}}(\bm{x}_{n})\,.
\end{align}
\begin{restatable}{prop}{theoremIJ}\label{theoremIJ}
	Under conditions \eqref{I0a}, \eqref{I0b} and
	\begin{align}\label{2}
		&\Im(\lambda_i - \lambda_j) = \Im(\gamma_i - \gamma_j) = 0,
		\qquad |\Im(\lambda_i - \gamma_j)| < \frac{\nu_g}{2}
	\end{align}
	the integrals \eqref{I}, \eqref{J} are absolutely convergent and explicit expressions
	are given by the formulae
	\begin{align}\label{I1}
		& I(\bm{\gamma}_{n}\,,\bm{\lambda}_{n+1}; x) =
		\,\frac{1}{d_n(g|\bo)}\,e^{2\pi\imath\left(\bbl_{n+1}-
			\bbg_n\right) x}\,
		\prod_{i=1}^{n+1}\prod_{j=1}^n\,\hat{K}(\lambda_i-\gamma_j)
		\,, \\[6pt]
		\label{J1}
		& J(\bm{\gamma}_{n}\,,\bm{\lambda}_{n}; x ) =
		\,\frac{1}{d_n(g|\bo)}\,e^{2\pi\imath\left(\bbl_{n}-
			\bbg_n\right) x }\,
		\prod_{i,j=1}^n \hat{K}(\lambda_i-\gamma_j)\,.
	\end{align}
\end{restatable}
Moreover, the integral \rf{J} delivers a delta sequence needed for the proof of Theorem~\ref{T-theorem}. Namely, set
\begin{align} \label{dseq2}
	\left( \Psi_{\bm{\lambda}'_n} |
	\Psi_{\bm{\lambda}_n} \right)_{\mu}^{x, \epsilon} =
	\int_{\R^{n}} d\bx_{n} \,
	\mu(\bx_{n})\,e^{2\pi \left(\varepsilon-\frac{\hat{g}}{2}\right)
	 (\underline{\bm{x}}_n-nx)}
	\,\K(x,\bx_{n})\,
	{\Psi}_{\bm{\lambda}'_n}(-\bm{x}_n)\,
	\Psi_{\bm{\lambda}_{n}}(\bm{x}_{n}).
\end{align}
\begin{restatable}{prop}{propdelta}\label{proposition-delta}
	Under conditions \rf{I0a}, \rf{I0b} and \rf{2} the following identity holds
	\begin{align}
		\lim_{x \rightarrow +\infty} \, \lim_{\epsilon \rightarrow 0+}
		\left( \Psi_{\bm{\lambda}'_n} | \Psi_{\bm{\lambda}_n} \right)_{\mu}^{x, \epsilon}
		= \frac{1}{\hat{\mu}(\bm{\lambda}_{n})}\,
		\delta (\bm{\lambda}_{n}, \bm{\lambda}'_{n})
	\end{align}
	in a sense of distributions.
\end{restatable}
Taking in mind the latter statements, we conjecture that bilinear analogs
\begin{align}&\int_{\R^{n}} d\bx_{n}\,
	\mu(\bx_{n})\,
	{\Psi}_{\bm{\lambda}'_n}(-\bm{x}_n)
	\Psi_{\bm{\lambda}_{n}}(\bm{x}_{n}) =
	\hat{\mu}^{-1}(\bm{\lambda}_{n})\,\delta (\bm{\lambda}'_{n}, \bm{\lambda}_{n}),\\[6pt]
	&\int_{\R^{n}} d\bl_{n}\,
	\hat{\mu}(\bl_{n})\,
	{\Psi}_{-\bm{\lambda}_n}(\bm{x}'_n)
	\Psi_{\bm{\lambda}_{n}}(\bm{x}_{n}) =
	\mu^{-1}(\bm{x}_{n})\,\delta (\bm{x}'_{n}, \bm{x}_{n})
\end{align}
of the relations \rf{St1} and \rf{St2} are valid for complex periods and complex coupling constant. However our analytical abilities are not sufficient for their proof so far.

\setcounter{equation}{0}
\section{Basic integrals}

In this section we calculate the integrals \eqref{I}, \eqref{J}, that is we prove Proposition \ref{theoremIJ}. Note that in the first integral we single out the dependence on $x = x_{n+1}$
in eigenfunction $\Psi_{\bm{\lambda}_{n+1}}(\bm{x}_{n+1}) = \Psi_{\bm{\lambda}_{n+1}}(\bm{x}_{n}\,, x)$ in order to emphasize that $x$ plays the role of external parameter which is not integrated.
 The  calculation of the scalar product of two eigenfunctions and
the calculation of the integral kernels of raising operators and $Q$-operators
in dual space can be actually reduced to the calculation of integrals \rf{I}, \rf{J}.

The iterative construction of eigenfunctions allows
to derive the recurrence relations for considered integrals, and it is the main
step of the proof. The derivation is heavily based on the fact that
integrals \eqref{I}, \eqref{J} contain $Q$-operator in a hidden form.
There are the following recurrence formulae connecting basic integrals
\begin{align}\label{IJ1}
& I(\bm{\gamma}_{n}\,,\bm{\lambda}_{n+1}; x ) =
\hat{K}(\lambda_{n+1}\,,\bm{\gamma}_{n})\,
e^{2\pi\imath\lambda_{n+1} x}\,
J(\bm{\gamma}_{n}\,,\bm{\lambda}_{n}; x )\,, \\[6pt]
\label{IJ2}
& J(\bm{\gamma}_{n}\,,\bm{\lambda}_{n}; x) =
\frac{d_{n-1}(g|\bo)}{d_n(g|\bo)}\,\hat{K}(\lambda_{n}\,,\bm{\gamma}_{n})\,
e^{2\pi\imath \lambda_{n}{x}}\,
I(\bm{\lambda}_{n-1}\,,\bm{\gamma}_{n}; x)\,.
\end{align}
This system has iterative structure and can be solved using \cbv the \cbk initial condition
\begin{align}\label{initial}
J(\gamma\,,\lambda; x) =
\int_{\R} dy\, e^{-2\pi \imath \gamma y}\,
\K(x - y)\,e^{2\pi \imath \lambda y } = \frac{1}{d_1(g|\bo)}\,
\hat{K}(\lambda-\gamma)\,e^{2\pi\imath(\lambda-\gamma) x}\,.
\end{align}
It is equivalent to the Fourier transform of the function $K(x)$ \eqref{K-fourier}. Then by induction the general solution of the system \eqref{IJ1}, \eqref{IJ2} is given by the formulae
\begin{align}
& I(\bm{\gamma}_{n}\,,\bm{\lambda}_{n+1}; x) =
\,\frac{1}{d_n(g|\bo)}\,e^{2\pi\imath\left(\underline{\bm{\lambda}}_{n+1}-
\underline{\bm{\gamma}}_n\right) x}\,
\prod_{i=1}^{n+1}\prod_{j=1}^n\,\hat{K}(\lambda_i-\gamma_j)
\,, \\[4pt]
& J(\bm{\gamma}_{n}\,,\bm{\lambda}_{n}; x ) =
\,\frac{1}{d_n(g|\bo)}\,e^{2\pi\imath\left(\underline{\bm{\lambda}}_{n}-
\underline{\bm{\gamma}}_n\right) x }\,
\prod_{i,j=1}^n \hat{K}(\lambda_i-\gamma_j)\,.
\end{align}

\subsection{ Recurrence relations}

Let us start from the derivation of \eqref{IJ1}.  Assuming $x_{n + 1} = x$ in definition \eqref{I} we have
\begin{multline*}
I(\bm{\gamma}_{n}\,,\bm{\lambda}_{n+1}; x) =
\int_{\R^{n}}d\bm{x}_{n}\,\mu(\bx_{n})\, \Psi_{-\bm{\gamma}_{n}}(\bm{x}_{n})\times\\[4pt]
d_n(g|\bo)\,\int_{\R^{n}}d\by_{n}\,\mu(\by_{n})\,
e^{2\pi\imath\lambda_{n+1}\left(\underline{\bm{x}}_{n+1}-\underline{\bm{y}}_{n}\right)}
\,\K(\bx_{n+1},\by_{n})\,
\Psi_{\bm{\lambda}_{n}}(\bm{y}_{n}) = \\[4pt]
e^{2\pi\imath\lambda_{n+1} x}\,\int_{\R^{n}}d\by_{n}\,\mu(\by_{n})\,
\K(x,\by_{n})\,\Psi_{\bm{\lambda}_{n}}(\bm{y}_{n})\,\times \\[4pt]
d_n(g|\bo)\,\int_{\R^{n}}d\bm{x}_{n}\,\mu(\bx_{n})\, \Psi_{-\bm{\gamma}_{n}}(\bm{x}_{n})\,
e^{2\pi\imath(-\lambda_{n+1})\left(\underline{\bm{y}}_{n}-\underline{\bm{x}}_{n}\right)}
\K(\bx_{n},\by_{n}) = \\[4pt]
\hat{K}(-\lambda_{n+1}\,,-\bm{\gamma}_{n})\,
e^{2\pi\imath\lambda_{n+1}{x}}\,
\int_{\R^{n}}d\by_{n}\,\mu(\by_{n})\,
\Psi_{-\bm{\gamma}_{n}}(\bm{y}_{n})\,
\K(x,\by_{n})\,\Psi_{\bm{\lambda}_{n}}(\bm{y}_{n})\,.
\end{multline*}
First, we used the relation
$\Psi_{\bm{\gamma}_{n}}(-\bm{x}_{n}) =
\Psi_{-\bm{\gamma}_{n}}(\bm{x}_{n})$ and the iterative formula for the eigenfunction $\Psi_{\bm{\lambda}_{n+1}}(\bm{x}_{n+1})$. Then in $\bm{x}_{n}$-integral we recognize the formula for the action of the operator $Q_n(-\lambda_{n+1})$ on the eigenfunction
$\Psi_{-\bm{\gamma}_{n}}(\bm{x}_{n})$. According to \eqref{QQPsi1} it produces the eigenvalue $\hat{K}(-\lambda_{n+1}\,,-\bm{\gamma}_{n})$. The last expression coincides with the formula \eqref{IJ1} due to the symmetry of $\hat{K}$-functions.

Derivation of the formula \eqref{IJ2} is based on the same steps: on the first step we use the iterative structure of eigenfunctions and on the second step in $\bm{y}_{n}$-integral we
recognize action of the operator $Q_{n}(-\lambda_{n})$ on the corresponding eigenfunction
\begin{align*}
&J(\bm{\gamma}_{n}\,,\bm{\lambda}_{n}; x) = \int_{\R^{n}}d\bm{y}_{n}\,\mu(\by_{n})
\, \Psi_{-\bm{\gamma}_{n}}(\bm{y}_{n})\,
\K(x,\by_{n})\,\Psi_{\bm{\lambda}_{n}}(\bm{y}_{n}) = \\[6pt] &\int_{\R^{n}}d\by_{n}\,\mu(\by_{n})\,\Psi_{-\bm{\gamma}_{n}}(\bm{y}_{n})\,
\K(x,\by_{n})\,\times \\[6pt]
&d_{n-1}(g|\bo)\,\int_{\R^{n-1}}d\bm{z}_{n-1}\,\mu(\bz_{n-1})\,
e^{2\pi\imath\lambda_{n}\left(\underline{\bm{y}}_{n}-\underline{\bm{z}}_{n-1}\right)}
\K(\bz_{n-1},\by_{n})\,
\Psi_{\bm{\lambda}_{n-1}}(\bm{z}_{n-1}) = \\[6pt]
&e^{2\pi\imath \lambda_{n} x}\,
\frac{d_{n-1}(g|\bo)}{d_n(g|\bo)}\,
\int_{\R^{n-1}}d\bz_{n-1}\,\mu(\bz_{n-1})\,
\Psi_{\bm{\lambda}_{n-1}}(\bm{z}_{n-1})\,\times \\[6pt]
& d_n(g|\bo)\,\int_{\R^{n}}d\bm{y}_{n}\,\mu(\by_{n})\,
\Psi_{-\bm{\gamma}_{n}}(\bm{y}_{n})\,
e^{2\pi\imath(-\lambda_{n})
\left(\underline{\bm{z}}_{n}-\underline{\bm{y}}_{n}\right)}
\K(\bz_{n},\by_{n}) = \\[6pt]
& \hat{K}(-\lambda_{n}\,,-\bm{\gamma}_{n})\,
e^{2\pi\imath \lambda_{n} x}\,
\frac{d_{n-1}(g|\bo)}{d_n(g|\bo)}\,
\int_{\R^{n-1}}d\bz_{n-1}\,\mu(\bz_{n-1})\,
\Psi_{\bm{\lambda}_{n-1}}(\bm{z}_{n-1})
\Psi_{-\bm{\gamma}_{n}}(\bm{z}_{n})\,,
\end{align*}
where we joined variables $x$ and $\bm{z}_{n-1}$
in one tuple $\bz_n = (\bz_{n-1},x)$.

Note that during the calculations we have interchanged the initial order of integrals and all these steps should be justified.

\subsection{ Convergence of basic integrals }

First, consider the integral
\begin{equation}
J(\bm{\gamma}_{n}\,,\bm{\lambda}_{n}; x) = \int_{\R^{n}}d\bm{x}_{n}\,\mu(\bx_{n})\,
\Psi_{\bm{\gamma}_{n}}(-\bm{x}_{n})\,
\K(x,\bx_{n})\,\Psi_{\bm{\lambda}_{n}}(\bm{x}_{n}).
\end{equation}
Denote by $\bm{y}_k, \bz_k$ vectors with $k$ components
\begin{equation}
\bm{y}_k = \bigl( y_1^{(k)}, \ldots, y_k^{(k)} \bigr), \qquad \bm{z}_k = \bigl( z_1^{(k)}, \ldots, z_k^{(k)} \bigr)
\end{equation}
and write the integral $J(\bg_n, \bl_n;x)$ in its full form using
\begin{equation}\label{Psi-}
\Psi_{\bm{\gamma}_n}(-\bm{x}_n) = \Psi_{-\bm{\gamma}_n}(\bm{x}_n)
\end{equation}
together with iterative representation of eigenfunctions \eqref{I12}
\begin{equation}\label{Jfull}
\begin{aligned}
J &= C_{J} \int d\bm{x}_n \prod_{j = 1}^{n - 1} \bigl[ d\bm{y}_j \, d\bm{z}_j \bigr] \, \mu(\bm{x}_n) \, K(x, \bm{x}_n) \, e^{2\pi \imath (\lambda_n - \gamma_n) \bbx_n}  \\[6pt]
&\times K(\bm{x}_n, \bm{y}_{n - 1}) \, e^{2\pi \imath (\lambda_{n - 1} - \lambda_n) \bby_{n - 1}} \prod_{j = 2}^{n - 1} \mu(\bm{y}_j)  \, e^{2\pi \imath (\lambda_{j - 1} - \lambda_j) \bby_{j - 1}} \, K(\bm{y}_j, \bm{y}_{j - 1}) \\[6pt]
&\times K(\bm{x}_n, \bm{z}_{n - 1}) \, e^{2\pi \imath (\gamma_{n} - \gamma_{n - 1}) \bbz_{n - 1}} \prod_{j = 2}^{n - 1} \mu(\bm{z}_j)  \, e^{2\pi \imath (\gamma_{j} - \gamma_{j - 1}) \bbz_{j - 1}} \, K(\bm{z}_j, \bm{z}_{j - 1}).
\end{aligned}
\end{equation}
Here all integrals are over $\R$ and $C_J$ contains all constants $d_k$ from eigenfunctions.

Let us prove that the integral \eqref{Jfull} is absolutely convergent assuming
\begin{equation}\label{assump}
\Im(\lambda_i - \lambda_j) = \Im(\gamma_i - \gamma_j) = 0, \qquad |\Im(\lambda_i - \gamma_j)| = \frac{\nu_g}{2} \delta < \frac{\nu_g}{2}
\end{equation}
for all $i,j$. Denote the integrand of \eqref{Jfull} by $F$. Also define function $S_n$ by the following recurrence relation
\begin{equation}
\begin{aligned}
S_n (\bm{y}_1, \dots, \bm{y}_n ) = \sum_{\substack{i, j = 1 \\ i \not= j}}^n \bigl| y_i^{(n)} - y_j^{(n)} \bigr| &- \sum_{i = 1}^{n} \sum_{j = 1}^{n - 1} \, \bigl| y_i^{(n)} - y_j^{(n - 1)} \bigr| \\
&+ S_{n - 1} (\bm{y}_1, \dots,  \bm{y}_{n - 1} )
\end{aligned}
\end{equation}
with $S_1 = 0$. Using bounds \eqref{Kmu-bound}
\begin{equation}\label{mu-K}
|\mu(x) | \leq C e^{\pi \nu_g |x|}, \qquad |K(x)| \leq C e^{- \pi \nu_g |x|}
\end{equation}
and triangle inequalities $|x-x_i| \geq |x_i| - |x|$ we have the estimate for the integrand
\begin{equation}
\begin{aligned}
|F| \leq C_F \exp \pi \nu_g \Biggl[ &(\delta - 1) \| \bm{x}_n \| - \sum_{\substack{i, j = 1 \\ i \not= j}}^n | x_i - x_j | \\[6pt]
&+ S_n(\bm{y}_1, \ldots, \bm{y}_{n - 1}, \bm{x}_n) + S_n(\bm{z}_1, \ldots, \bm{z}_{n - 1}, \bm{x}_n)  \Biggr]
\end{aligned}
\end{equation}
with some $C_F(x, g, \bm{\omega})$. Here and in what follors we denote
\begin{equation}
\| \bx_n \| = |x_1 | + \ldots + |x_n|.
\end{equation}
In our previous paper we proved the following bound for $S_n$
\begin{equation}\label{S-ineq}
S_n (\bm{y}_1, \dots, \by_{n - 1}, \bm{x}_n ) \leq \frac{1}{2} \sum_{\substack{i, j = 1 \\ i \not= j}}^n \bigl| x_i - x_j \bigr| + \epsilon \| \bm{x}_n \| - \frac{\epsilon}{(n - 1)! e} \sum_{k = 1}^{n - 1} \| \bm{y}_k \|
\end{equation}
that holds for any $\eps \in [0, 2(n - 1)]$, see \cite[Lemma 2]{BDKK2} (note that coefficients $c_n$ appearing in \textit{loc. cit.} are bounded as $c_n \leq (n - 1)! e$). Applying this inequality twice we arrive at
\begin{equation}
|F| \leq C_F \exp \pi \nu_g \biggl[ (2 \epsilon + \delta - 1) \| \bm{x}_n \| - \frac{\epsilon}{(n - 1)! e} \sum_{k=1}^{n - 1} ( \| \bm{y}_k \| + \| \bm{z}_k \| ) \biggr].
\end{equation}
Since $\delta < 1$, we can take small $\epsilon>0$, such that $2\epsilon + \delta - 1< 0$. Therefore, the integral $J$ is absolutely convergent.

Next, consider the integral
\begin{equation}
I(\bm{\gamma}_{n}\,,\bm{\lambda}_{n+1}; x) = \int_{\R^{n}}d\bm{x}_{n}\, \mu(\bx_{n})\,
\Psi_{\bm{\gamma}_{n}}(-\bm{x}_{n})\, \Psi_{\bm{\lambda}_{n+1}}(\bm{x}_{n},x)
\end{equation}
and write it in its full form using again \eqref{Psi-} and eigeinfunctions iterative representation~\eqref{I12}
\begin{equation}\label{Ifull}
\begin{aligned}
I &= C_{I} \, e^{2\pi \imath \l_{n + 1}x} \int d\bm{x}_n d\bm{y}_n \prod_{j = 1}^{n - 1} \bigl[ d\bm{y}_j \, d\bm{z}_j \bigr] \, \mu(\bm{x}_n) \, e^{2\pi \imath (\lambda_{n + 1} - \gamma_n) \bbx_n} \, K(x, \bm{y}_n)  \\[6pt]
&\times K(\bm{x}_n, \bm{y}_{n}) \, e^{2\pi \imath (\lambda_{n} - \lambda_{n + 1}) \bby_{n}} \prod_{j = 2}^{n} \mu(\bm{y}_j)  \, e^{2\pi \imath (\lambda_{j - 1} - \lambda_j) \bby_{j - 1}} \, K(\bm{y}_j, \bm{y}_{j - 1}) \\[6pt]
&\times K(\bm{x}_n, \bm{z}_{n - 1}) \, e^{2\pi \imath (\gamma_{n} - \gamma_{n - 1}) \bbz_{n - 1}} \prod_{j = 2}^{n - 1} \mu(\bm{z}_j)  \, e^{2\pi \imath (\gamma_{j} - \gamma_{j - 1}) \bbz_{j - 1}} \, K(\bm{z}_j, \bm{z}_{j - 1}).
\end{aligned}
\end{equation}
Here $C_I$ again contains all unimportant constants. Assuming the same conditions \eqref{assump} let us prove that this integral is absolutely convergent. For this we need one more inequality which we prove in Appendix \ref{AppB}. Define
\begin{equation}
R_n(\bm{y}_{n}, \bm{x}_{n}) = \sum_{\substack{i, j = 1 \\ i < j}}^{n} | x_i - x_j | + \sum_{\substack{i, j = 1 \\ i < j}}^{n} | y_i - y_j | - \sum_{i, j = 1}^{n} \, | x_i - y_j |.
\end{equation}
Then from Corollary \ref{cor:R-ineq} in Appendix \ref{AppB} we have for any $\eps \in[0,1]$
\begin{equation}\label{R-ineq}
R_n \leq \epsilon ( \|\bm{y}_n \| - \| \bm{x}_n \|).
\end{equation}

Denote the integrand of \eqref{Ifull} by $G$. Using bounds \eqref{mu-K} we arrive at the estimate
\begin{equation}
\begin{aligned}
|G| & \leq C_G \exp \pi \nu_g \Biggl[  \delta \| \bm{x}_n \| - \| \bm{y}_n \| + R_n(\bm{y}_n, \bm{x}_n) - \sum_{\substack{i, j = 1 \\ i < j}}^n | x_i - x_j |\\[6pt]
&- \sum_{\substack{i, j = 1 \\ i < j}}^n \bigl| y^{(n)}_i - y^{(n)}_j \bigr|
+ S_n(\bm{y}_1, \ldots, \bm{y}_n) + S_n(\bm{z}_1, \ldots, \bm{z}_{n - 1}, \bm{x}_n) \Biggr]
\end{aligned}
\end{equation}
with some $C_G(x, g, \bm{\omega})$. Using inequalities \eqref{S-ineq} with $\epsilon_1$ (twice) and \eqref{R-ineq} with $\epsilon_2$ we obtain
\begin{equation}
\begin{aligned}
|G|  \leq C_G \exp \pi \nu_g \Biggl[  (\delta + \epsilon_1 - \epsilon_2) \| \bm{x}_n \| & + (\epsilon_1 + \epsilon_2 - 1) \| \bm{y}_n \| \\[6pt]
&- \frac{\epsilon_1}{(n - 1)! e}\sum_{k = 1}^{n - 1}(\| \bm{y}_k \| + \| \bm{z}_{k} \| ) \Biggr].
\end{aligned}
\end{equation}
To have integrable function from the right we need to choose $\epsilon_1, \epsilon_2 > 0$, such that
\begin{equation}
\delta + \epsilon_1 - \epsilon_2 < 0, \qquad \epsilon_1 + \epsilon_2 - 1 < 0.
\end{equation}
Since $\delta \in [0,1)$, it is true, for example, for
\begin{equation}
\epsilon_1 = \frac{1 - \delta}{6}, \qquad \epsilon_2 = \frac{2 \delta + 1}{3}.
\end{equation}
Thus, the integral $I$ is absolutely convergent.

\setcounter{equation}{0}
\section{Orthogonality and completeness of  wave functions}
We fix $\omega_1\,,\omega_2$ and $g$ to be real  satisfying conditions \eqref{I0a}. Our aim is to show that the functions $\Psi_{\bm{\lambda}_n}(\bm{x}_n)$ form a complete orthogonal set in the Hilbert space $L^2_{\text{sym}}\left(\mathbb R^{n}, \mu\right)$ of symmetric functions of variables $\bm{x}_n$ with the scalar product
\begin{equation}
\langle \Phi | \Psi \rangle_{\mu} = \int d\bm{x}_n \, \mu(\bm{x}_n) \, \overline{\Phi(\bm{x}_n)} \, \Psi(\bm{x}_n).
\end{equation}
It is natural to introduce the dual Hilbert space $L^2_{\text{sym}}\left(\mathbb R^{n}, \hat{\mu}\right)$ of symmetric functions of variables $\bm{\lambda}_n$ with the scalar product
\begin{equation}
\langle \varphi | \psi \rangle_{\hat{\mu}} = \int d\bm{\lambda}_n \,
\hat{\mu}(\bm{\lambda}_n) \, \overline{\varphi(\bm{\lambda}_n)} \, \psi(\bm{\lambda}_n).
\end{equation}
The essential difference between the two Hilbert spaces is the difference in the measure $\mu \rightleftarrows \hat{\mu}$. To avoid misunderstanding we label each scalar product and the related norm
\begin{align}
&\| \Phi \|_{\mu}^2 = \int d\bm{x}_n \, \mu(\bm{x}_n) \,
|\Phi(\bm{x}_n)|^2 \,,\\
&\|{\varphi}\|_{\hat{\mu}}^2
= \int d\bm{\lambda}_n \, \hat{\mu}(\bm{\lambda}_n) \,
|\varphi(\bm{\lambda}_n)|^2\,
\end{align}
by the corresponding index $\mu$ or $\hat{\mu}$.

For real $\bm{\lambda}_n$ the functions $\Psi_{\bm{\lambda}_n}(\bm{x}_n)$ do not belong to the Hilbert space $L^2_{\text{sym}}\left(\mathbb R^{n}, \mu\right)$. However, they allow us  to
define a linear transform from the space of smooth function to the space of smooth functions on $\R^n$
\begin{align}\label{T-tr}
(T \varphi) (\bm{x}_n) = \int d\bm{\lambda}_n \, \hat{\mu}(\bm{\lambda}_n) \, \Psi_{\bm{\lambda}_n} (\bm{x}_n) \, \varphi(\bm{\lambda}_n)\,.
\end{align}

\Ttheorem*

Note that the relation $\| T \varphi\|_{\mu}^2=\|{\varphi}\|_{\hat{\mu}}^2$ can be extended by a standard polarization procedure to the relation
\begin{align}\label{phipsi}
\langle T \varphi| T \psi\rangle_{\mu} =
\langle \varphi|\psi\rangle_{\hat{\mu}}.
\end{align}
The latter identity is an exact formulation of the orthogonality relation
\begin{equation}\label{ort2}
	\int_{\R^{n}} d\bx_{n}\,
	\mu(\bx_{n})\,
	\overline{\Psi}_{\bm{\lambda}'_n}(\bm{x}_n)
	\Psi_{\bm{\lambda}_{n}}(\bm{x}_{n}) =
	\left(\hat{\mu}(\bm{\lambda}_{n})\right)^{-1}\,
	\delta (\bm{\lambda}'_{n}, \bm{\lambda}_{n})
\end{equation}
in a sense of distributions.

The proof of Theorem \ref{T-theorem} is given in Section \ref{sectiongeneral}. In Section \ref{sec:n1} we explain main steps of this proof on $n=1$ example.  The proof of Theorem \ref{S-theorem}, which is equivalent to completeness relation, is also given in Section \ref{sectiongeneral}.

\subsection{Fourier transform as $n=1$ example}\label{sec:n1}
In the simplest case $n=1$ the measure is trivial $\mu(x) = \hat{\mu}(\lambda) = 1$ and
the eigenfunction coincides with the plane wave.
Hence, the transform $T$ \eqref{T-tr} is the Fourier transform and it
is possible to illustrate the whole scheme step by step
using this example.
Namely, we prove that the system of functions
$\{e^{2\pi \imath \lambda x} , \lambda \in \mathbb R\}$
forms an orthogonal system in $L^2(\R )$.
For brevity we use standard notation $\hat{\varphi}(x)$ for the Fourier transform
instead of general notation $(T \varphi)(x)$. We have
\begin{align}
\hat{\varphi}(x) = \int_\R d \lambda \,
e^{2\pi \imath \lambda x} \, \varphi(\lambda)\,.
\end{align}
For smooth, compactly supported functions $\varphi$ on $\mathbb{R}$
the function $\hat{\varphi}(x)$ rapidly decreases as $|x| \to \infty$ (which can be seen by  integration by parts). This guarantees the convergence of the integral
\begin{align}
\|\hat{\varphi}\|^2 = \int_\R dx \, \overline{\hat{\varphi}(x)}\,\hat{\varphi}(x) =
\int_\R dx \, |\hat{\varphi}(x)|^2.
\end{align}
Due to this convergence it is possible to introduce regularization inside the integral using appropriate function $f_{\epsilon}(x)$ which has pointwise limit: $\lim\limits_{\epsilon\to 0+} f_{\epsilon}(x) = 1$, so that
\begin{align}
\|\hat{\varphi}\|^2 = \int_\R dx \, \lim_{\epsilon\to 0+}\, f_{\epsilon}(x)\,
|\hat{\varphi}(x)|^2 =
\lim_{\epsilon\to 0+}\,\int_\R dx \, f_{\epsilon}(x)\, |\hat{\varphi}(x)|^2.
\end{align}
The first step is trivial, and on the second step we use dominated convergence theorem which
states that the limit can be taken out of the integral, provided that there exists function $g(x)$, such that $|f_{\epsilon}(x)| \leq g(x)$ for all $\epsilon, x$ and the integral
\begin{equation}
\int_\R dx \, g(x) \, | \hat{\vf} (x)|^2
\end{equation}
is convergent.

The simplest example of regularization is the following
\begin{align}
\|\hat{\varphi}\|^2 = \int_\R dx \, \lim_{\varepsilon \to 0+}\,
e^{-\varepsilon\, |x|}\,|\hat{\varphi}(x)|^2 =
\lim_{\varepsilon \to 0+}\, \int_\R dx \,e^{-\varepsilon\, |x|}\,
|\hat{\varphi}(x)|^2.
\end{align}
The equality $\lim_{\varepsilon \to 0+}\,e^{-\varepsilon\, |x|} = 1$ is evident and the condition of dominated convergence theorem is fulfilled for $g(x)=1$ since $e^{-\varepsilon\, |x|} \leq 1$.

Let us consider a more complicated regularization which
will admit generalization to any $n$
\begin{align*}
\|\hat{\varphi}\|^2 = \int_\R dx \,
\lim_{z \rightarrow +\infty} \, \lim_{\epsilon \rightarrow 0+}
\,e^{\pi \hat{g}(z - x) + 2\pi \epsilon (x-z)} \,
K(z - x) \, |\hat{\varphi}(x)|^2 = \\[4pt]
\lim_{z \rightarrow +\infty} \, \lim_{\epsilon \rightarrow 0+}\,
\int_\R dx \,e^{\pi \hat{g}(z - x) + 2\pi \epsilon (x-z)} \,
K(z - x) \, |\hat{\varphi}(x)|^2\,.
\end{align*}
 Here we recall that
\begin{equation}
\hat{g} = \frac{g}{\o_1 \o_2} > 0.
\end{equation}
Using asymptotics of $K$-functions one may readily check that
\begin{equation}
\lim_{z \rightarrow + \infty} \, \lim_{\ve \rightarrow 0+} e^{\pi \hat{g}(z - x) + 2\pi \varepsilon (x-z)}
K(z - x) = 1
\end{equation}
pointwise.
Indeed, it is evident that $e^{2\pi \varepsilon (x-z)} \to 1$ as
$\ve \to 0$ and for $z\to +\infty$ we have~\eqref{Kmu-asymp}
\begin{align}
K(z - x)\to e^{\pi \hat{g} \left(x-z\right)}.
\end{align}
Next to use dominated convergence theorem recall inequality
\eqref{Kmu-bound}
$$
|K(y)| \leq C\,e^{-\pi\nu_g|y|}, \qquad  y \in \mathbb{R}
$$
where  $\nu_g = \Re \hat{g}  = \hat{g} >0$.  Therefore, we have
\begin{align}\label{C}
\left|e^{\pi \hat{g}(z - x) + 2\pi \epsilon (x-z)} \,
K(z - x)\right| \leq C \, e^{ \pi (\hat{g} - 2 \ve) (z - x) - \pi \hat{g} |z - x| } \leq C e^{- 2 \pi \ve | z - x|} \leq C,
\end{align}
where on the second step we used inequality $a \leq |a|$ and the fact that $2\ve < \hat{g}$\cbk.

The main goal of regularization is the following.
In the initial integral
\begin{align*}
\|\hat{\varphi}\|^2 = \int_\R dx \, \overline{\hat{\varphi}(x)}\,\hat{\varphi}(x) =
\int_\R dx \, \int_\R d \lambda' \,
e^{-2\pi \imath \lambda' x} \, \overline{\varphi(\lambda')}\,
\int_\R d \lambda \,
e^{2\pi \imath \lambda x} \, \varphi(\lambda)
\end{align*}
the order of integrations is fixed from the very beginning and
can not be changed. One can not invoke Fubini's theorem because the
$x$-integral is not absolute convergent separately.
Regularization improves absolute convergence of the
$x$-integral so that it is possible to invoke Fubini's theorem,
change order of integrations and then calculate the $x$-integral first to obtain delta sequence.
In the case of the first regularization we have
\begin{align*}
&\|\hat{\varphi}\|^2 = \int_\R dx \, \overline{\hat{\varphi}(x)}\,\hat{\varphi}(x) =
\lim_{\varepsilon \to 0+}\, \int_\R dx \,e^{-\varepsilon\, |x|}\, \int_\R d \lambda' \,
e^{-2\pi \imath \lambda' x} \, \overline{\varphi(\lambda')}\,
\int_\R d \lambda \,
e^{2\pi \imath \lambda x} \, \varphi(\lambda) = \\[6pt]
&\lim_{\varepsilon \to 0+}\, \int_{\R^2} d \lambda' \, d \lambda \,
\overline{\varphi(\lambda')}\,\varphi(\lambda)\,
\int_\R dx \,e^{-\varepsilon\, |x|}\,e^{2\pi \imath (\lambda-\lambda') x} = \\[6pt]
&\lim_{\varepsilon \to 0+}\, \int_{\R^2} d \lambda' \, d \lambda \,
\overline{\varphi(\lambda')}\,\varphi(\lambda)\,\frac{1}{2\pi \imath}\,
\left[\frac{1}{\lambda-\lambda'-\imath\epsilon}-\frac{1}{\lambda-\lambda'+\imath\epsilon}\right] =
\int_\R d \lambda \,
\overline{\varphi(\lambda)}\,\varphi(\lambda)\,.
\end{align*}
 In the last line we rescaled $\ve \rightarrow 2\pi\ve$. In fact, the proof of the last equality is the same as the proof of the well-known formula in the theory of the generalized functions
\begin{align}
\lim_{\varepsilon \to 0+} \left[\frac{1}{\lambda-\lambda'- \imath\epsilon}-\frac{1}{\lambda-\lambda'+\imath\epsilon}\right] =
2\pi \imath \,\delta(\lambda-\lambda').
\end{align}
Thus, we proved the equality $\|\hat{\varphi}\|^2 = \|\varphi\|^2$ and as a consequence
the orthogonality relation in the form
\begin{align}
\int_\R dx \,e^{2\pi \imath (\lambda-\lambda') x} = \delta(\lambda-\lambda')
\end{align}
where integral is understood in the sense of generalized functions as described above.

Let us consider the second regularization as an illustrative example to the general case of arbitrary $n$
\begin{align*}
&\|\hat{\varphi}\|^2 =
\lim_{z \rightarrow +\infty} \, \lim_{\epsilon \rightarrow 0+}\,
\int_\R dx \,e^{\pi \hat{g}(z - x) + 2\pi \epsilon (x-z)} \,
K(z - x) \, |\hat{\varphi}(x)|^2 = \\[6pt]
&\lim_{z \rightarrow +\infty} \, \lim_{\epsilon \rightarrow 0+}\,
\, \int_{\R^2} d \lambda' \, d \lambda \,
\overline{\varphi(\lambda')}\,\varphi(\lambda)\,
\int_\R dx \,e^{\pi \hat{g}(z - x) + 2\pi \epsilon (x-z)} \, K(z - x)
\,e^{2\pi \imath (\lambda-\lambda') x}.
\end{align*}
The last integral coincides with the simplest basic integral $J$ from the previous section \eqref{initial} and can be calculated in a closed form
\begin{align*}
\int_\R dx \,e^{\pi \hat{g}(z - x) + 2\pi \epsilon (x-z)} \, K(z - x)
\,e^{2\pi \imath (\lambda-\lambda') x} = \frac{1}{d_1(g|\bo)}\,
\hat{K}(\lambda-\lambda' +\textstyle \frac{\imath \hat{g}}{2} -\imath\epsilon)\,
e^{2\pi\imath(\lambda-\lambda') z}.
\end{align*}
In Section \ref{sectiongeneral} we show that the last expression is a delta sequence
\begin{align}
\lim_{z \rightarrow +\infty} \, \lim_{\epsilon \rightarrow 0+}\,
\hat{K}(\lambda-\lambda' +\textstyle \frac{\imath \hat{g}}{2} -\imath\epsilon)\,
e^{2\pi\imath(\lambda-\lambda') z} = d_1(g|\bo)\, \delta(\lambda-\lambda')\,.
\end{align}

In the case of Fourier transform we are able to use the
well-known fact that for smooth, compactly supported functions
$\varphi$ on $\mathbb{R}$ the function $\hat{\varphi}(x)$ rapidly
decreases as $|x| \to \infty$ that guarantees the convergence
of the integral
\begin{align}
\|\hat{\varphi}\|^2 = \int_\R dx \, |\hat{\varphi}(x)|^2
\end{align}
and then prove the equality
\begin{align}
\|\hat{\varphi}\|^2 = \lim_{\epsilon\to 0+}\,\int_\R dx \, f_{\epsilon}(x)\, |\hat{\varphi}(x)|^2
\end{align}
using dominated convergence theorem and corresponding restriction on $f_{\epsilon}(x)$.  In general case we do not know the same property of rapid decay for the transformed function $T \vf$, although it is probably true even in the case of complex periods and coupling constant\cbk.
So, in general case we slightly change the logic and prove convergence of the needed
integral using Fatou's lemma. The corresponding steps in the case of Fourier transformation are as follows.
\begin{itemize}
\item At the first step we consider regularized integral  with $f_\ve(x) \geq 0$  and prove that
\begin{align}
\lim_{\epsilon\to 0+}\,\int_\R dx \, f_{\epsilon}(x)\, |\hat{\varphi}(x)|^2 =
\int_\R d\lambda \, |\varphi(\lambda)|^2.
\end{align}
\item At the second step we use Fatou's lemma which states that
\begin{align}
\int_\R dx \, \varliminf_{\epsilon\to 0+}\, f_{\epsilon}(x)\, |\hat{\varphi}(x)|^2 \leq \varliminf_{\epsilon\to 0+}\,\int_\R dx \, f_{\epsilon}(x)\, |\hat{\varphi}(x)|^2 =
\int_\R d\lambda \, |\varphi(\lambda)|^2
\end{align}
and therefore the needed integral is convergent
\begin{align}
\int_\R dx \,|\hat{\varphi}(x)|^2 \leq \int_\R d\lambda \, |\varphi(\lambda)|^2.
\end{align}
\item Convergence of the needed integral allows us to use dominated convergence theorem in the same way as before, so that
\begin{align}
\|\hat{\varphi}\|^2 = \int_\R dx \, \lim_{\epsilon\to 0+}\, f_{\epsilon}(x)\,
|\hat{\varphi}(x)|^2 =
\lim_{\epsilon\to 0+}\,\int_\R dx \, f_{\epsilon}(x)\, |\hat{\varphi}(x)|^2
\end{align}
and therefore $\|\hat{\varphi}\|^2 = \|\varphi\|^2$.
\end{itemize}

\subsection{General case} \label{sectiongeneral}
In order to prove Theorem \ref{T-theorem}, we first establish that the identity \eqref{unit-T}
\begin{equation}
\| T \vf \|^2_\mu = \| \vf \|^2_{\hat{\mu}}
\end{equation}
holds for smooth, compactly supported functions $\varphi$ on $\mathbb{R}^{n}$.
Let us introduce the regularised scalar product
\begin{align}
\langle T\varphi | T\varphi\rangle^{x,\epsilon}_{\mu} =
\int_{\R^n} d\bm{x}_n \, \mu(\bm{x}_n)\,
 e^{\pi \hat{g} (n x - \bbx_n) + 2\pi \epsilon (\bbx_n-nx)} \,
K(x, \bm{x}_n) \, |T\varphi(\bm{x}_n)|^2.
\end{align}
Note that for real parameters $K(x, \bm{x}_n) \geq 0$.
Using asymptotics of $K$-functions one may readily check the first
needed property of regularization
\begin{align}
\lim_{x \rightarrow +\infty} \, \lim_{\epsilon \rightarrow 0+}
e^{\pi \hat{g} (n x - \bbx_n) + 2\pi \epsilon (\bbx_n-nx)} \,
K(x, \bm{x}_n) = 1.
\end{align}
Indeed, it is evident that $e^{\pi \varepsilon (\underline{\bm{x}}_n -nx)} \to 1$ as
$\ve \to 0$ and as $x\to +\infty$ we have \eqref{Kmu-asymp}
\begin{align}
K(x, \bx_n) = \prod_{i=1}^{n} K(x_i-x)\to
\prod_{i=1}^{n} e^{\pi \hat{g} (x_i-x)} = e^{\pi \hat{g} \left(\underline{\bm{x}}_n -n x\right)}.
\end{align}
Now let us show that the regularized integral
\begin{equation}\label{Tphi-full}
\begin{aligned}
\langle T\varphi | T\varphi\rangle^{x,\epsilon}_{\mu} = \int_{\R^{3n}} d\bm{x}_n  \, d\bl_n' \, & d\bl_n  \;\, \mu(\bm{x}_n)\, e^{\pi \hat{g} (n x - \bbx_n) + 2\pi \epsilon (\bbx_n-nx)} \,
K(x, \bm{x}_n) \\[6pt]
& \hspace{0.35cm} \times \hat{\mu}(\bl_n') \, \hat{\mu}(\bl_n) \;  \overline{\Psi}_{\bl_n'}(\bx_n) \, \Psi_{\bl_n}(\bx_n)  \;  \overline{\vf(\bl_n') } \, \vf(\bl_n)
\end{aligned}
\end{equation}
is absolutely convergent. The integral over $\bx_n$
\begin{align}\label{dseq4}
\langle \Psi_{\bm{\lambda}'_n} |
\Psi_{\bm{\lambda}_n} \rangle_{\mu}^{x, \epsilon} =
\int_{\R^{n}} d\bx_{n}\,
\mu(\bx_{n})
e^{2\pi \left(\varepsilon-\frac{g}{2}\right) (\underline{\bm{x}}_n-nx)}
\,\K(x,\bx_{n})\,
\overline{\Psi}_{\bm{\lambda}'_n}(\bm{x}_n)\,
\Psi_{\bm{\lambda}_{n}}(\bm{x}_{n})
\end{align}
almost coincides with the basic integral $J$ \eqref{J}
\begin{equation}\label{ref-product0}
\langle \Psi_{\bm{\lambda}'_n} |
\Psi_{\bm{\lambda}_n} \rangle_{\mu}^{x, \epsilon} = e^{\pi \hat{g} n x} \,
e^{-2\pi \varepsilon n x}\,
J \left( \bm{\lambda}'_n \,, \bm{\lambda}_n + \bigl(\textstyle\frac{\imath \hat{g}}{2} -
\imath \epsilon\bigr) \bm{e}\, ; x \right).
\end{equation}
Here we use shorthand notation
$$
\bm{e} = (1,1,\ldots,1), \qquad
\bm{\lambda}_{n}+\lambda\bm{e} = (\lambda_1+\lambda,\ldots,\lambda_n+\lambda)
$$
and the following relation
$$
e^{2 \pi \imath \lambda\,\underline{\bm{x}}_n}\,\Psi_{\bm{\lambda}_{n}}(\bm{x}_{n}) =
\Psi_{\bm{\lambda}_{n}+\lambda\bm{e}}(\bm{x}_{n})\,,
$$
which can be easily checked.
As we proved in the previous section, the integral $J$ is absolutely convergent. Since $\vf(\bl_n)$ is compactly supported, the whole integral \eqref{Tphi-full} is also absolutely convergent.

Next we prove the following two equalities
\begin{equation}\label{T}
\| T\varphi \|_{\mu}^2 = \lim_{x \rightarrow +\infty} \, \lim_{\epsilon \rightarrow 0+}
\langle T\varphi | T\varphi\rangle^{x,\epsilon}_{\mu} = \| \varphi \|_{\hat{\mu}}^2\,.
\end{equation}
Let us consider the second one
\begin{equation}\label{Treg}
\lim_{x \rightarrow +\infty} \, \lim_{\epsilon \rightarrow 0+}
\langle T\varphi | T\varphi\rangle^{x,\epsilon}_{\mu} =  \| \varphi\|_{\hat{\mu}}^2.
\end{equation}
The proof of this relation is divided into two steps.
In the first step we invoke Fubini's theorem and interchange the order of integrals in
\eqref{Tphi-full} evaluating $\bm{x}_n$-integral first
\begin{align*}
\lim_{x \rightarrow +\infty} \, \lim_{\epsilon \rightarrow 0+}
\langle T\varphi | T\varphi\rangle^{x,\epsilon}_{\mu} =
\lim_{x \rightarrow +\infty} \, \lim_{\epsilon \rightarrow 0+}
\int d\bl_n' \, d\bl_n \, \hat{\mu}(\bl_n') \, \hat{\mu}(\bl_n) \;
\overline{\vf(\bl_n')} \, \vf(\bl_n)  \,
\langle \Psi_{\bm{\lambda}'_n} | \Psi_{\bm{\lambda}_n} \rangle_{\mu}^{x, \epsilon}.
\end{align*}
In the second step we use explicit expression for $\langle \Psi_{\bm{\lambda}'_n} | \Psi_{\bm{\lambda}_n} \rangle_{\mu}^{x, \epsilon}$ and prove that
\begin{align}\label{dseq1}
\lim_{x \rightarrow +\infty} \, \lim_{\epsilon \rightarrow 0+}
\langle \Psi_{\bm{\lambda}'_n} | \Psi_{\bm{\lambda}_n} \rangle_{\mu}^{x, \epsilon}
 = \frac{1}{\hat{\mu}(\bm{\lambda}_{n})}\,
\delta (\bm{\lambda}_{n}, \bm{\lambda}'_{n})\,.
\end{align}
Finally, $\bm{\lambda}$-integrals produce needed result $\| \varphi\|_{\hat{\mu}}^2$.

Actually, we prove \rf{dseq1} in a more general setting for bilinear pairing \rf{Mpair} and complex valued periods and coupling constant, so that the statement \rf{dseq1} will be its particular case.

\propdelta*

Here the regularized pairing from the left is defined in \rf{dseq2}
\begin{align}
\left( \Psi_{\bm{\lambda}'_n} |
\Psi_{\bm{\lambda}_n} \right)_{\mu}^{x, \epsilon} =
\int_{\R^{n}} d\bx_{n} \,
\mu(\bx_{n})\,e^{2\pi \left(\varepsilon-\frac{\hat{g}}{2}\right)
	(\underline{\bm{x}}_n-nx)}
\,\K(x,\bx_{n})\,
{\Psi}_{\bm{\lambda}'_n}(-\bm{x}_n)\,
\Psi_{\bm{\lambda}_{n}}(\bm{x}_{n}).
\end{align}
As well as the regularized scalar product \rf{dseq4} it converges and is connected to the integral $J$ by the same relation \rf{ref-product0}, so that it can be evaluated in a closed form \eqref{J1}
\begin{equation}\label{ref-product}
\begin{aligned}
\left( \Psi_{\bm{\lambda}'_n} |
\Psi_{\bm{\lambda}_n} \right)_{\mu}^{x, \epsilon} &= e^{\pi \hat{g} n x} \,
e^{-2\pi \varepsilon n x}\,
J \left( \bm{\lambda}'_n \,, \bm{\lambda}_n + \bigl(\textstyle\frac{\imath \hat{g}}{2} -
\imath \epsilon\bigr) \bm{e}\, ; x \right) \\[6pt]
& = \frac{1}{d_n(g|\bo)}\,
e^{2\pi\imath\left(\underline{\bm{\lambda}}_{n}
	-\underline{\bm{\lambda}}'_n\right) x}\,
\prod_{i,j=1}^n \hat{K} \Bigl(\lambda_i-\lambda'_j+
\frac{\imath \hat{g}}{2} - \imath\varepsilon \Bigr).
\end{aligned}
\end{equation} \cbk
Next we rewrite this expression using definitions \eqref{I6c}, \eqref{I5d} and reflection formula~\eqref{inv}
\begin{multline*}
\hat{K} \Bigl(\lambda_i-\lambda'_j+
\frac{\imath \hat{g}}{2} - \imath\varepsilon \Bigr) = \\[4pt]
S_2^{-1}\Bigl(\imath (\lambda_i-\lambda'_j-\imath\varepsilon)\Big|\hat{\bo}\Bigr)\,
S_2^{-1}
\Bigl(-\imath (\lambda_i-\lambda'_j-\imath\varepsilon)+\hat{g}\Big|\hat{\bo}\Bigr) =
\frac{1}{\mm(\lambda_i-\lambda'_j-\imath\varepsilon)}.
\end{multline*}
Collecting everything together we obtain the following expression
for the regularized scalar product
\begin{align*}
\left( \Psi_{\bm{\lambda}'_n} |
\Psi_{\bm{\lambda}_n} \right)_{\mu}^{x, \epsilon} =
\frac{1}{d_n(g|\bo)}\,
\frac{e^{2\pi\imath\left(\underline{\bm{\lambda}}_{n} -
\underline{\bm{\lambda}}'_{n}\right)x }}
{\prod_{i,j=1}^n
S_2\left(\imath (\lambda_i-\lambda'_j-\imath\varepsilon)|\hat{\bo}\right)\,
S_2
\left(-\imath (\lambda_i-\lambda'_j-\imath\varepsilon)+\hat{g}|\hat{\bo}\right)}.
\end{align*}
On the last step we calculate the limit
\begin{align*}
\lim _{x \to +\infty}\,
\lim _{\varepsilon \to 0+}\,e^{2\pi \imath
(\underline{\bm{\lambda}}_n-\underline{\bm{\lambda}}'_n)x}\,
\prod_{i,k=1}^n \frac{1}{S_2\left(\imath(\lambda'_i - \lambda_k +\imath\varepsilon)+\hat{g}|\hat{\bo}\right)
S_2\left(\imath(\lambda_k - \lambda'_i -\imath\varepsilon)|\hat{\bo}\right)}
\end{align*}
using the following Lemma \cite{DKM} (see Appendix \ref{App-delta} for details).

\begin{lemma}\label{lemma:delta}
In the sense of distributions the following identity holds
\begin{align*}
\lim_{\lambda\rightarrow +\infty} \, \lim_{\epsilon\rightarrow 0+}
\frac{e^{2\pi\imath \lambda(\underline{\bm{x}}_n-\underline{\bm{y}}_n)}}{
\prod_{a,b=1}^{n} \left(x_a-y_b -\imath\varepsilon\right)}
= \frac{(-1)^{\frac{n(n-1)}{2}}(2\pi \imath )^{n} n!}
{\prod_{a<b}^{n} \left(x_{a}-x_{b}\right)^2}\,
\delta(\bm{x}_n\,,\bm{y}_n ),
\end{align*}
where
\begin{align*}
\delta(\bm{x}_n\,,\bm{y}_n ) = \frac1{n!}\,
\sum_{w\in S_n} \prod_{k=1}^n\delta(x_k-y_{w(k)}).
\end{align*}
\end{lemma}

We have
\begin{align*}
&\lim _{x \to +\infty}\,
\lim _{\varepsilon \to 0+}\,e^{2\pi \imath
(\underline{\bm{\lambda}}_n-\underline{\bm{\lambda}}'_n) x}\,
\prod_{i,k=1}^n \frac{1}{S_2\left(\imath(\lambda'_i - \lambda_k +\imath\varepsilon)+\hat{g}|\hat{\bo}\right)
S_2\left(\imath(\lambda_k - \lambda'_i -\imath\varepsilon)|\hat{\bo}\right)} = \\[6pt]
&\lim _{\varepsilon \to 0+}\,
\prod_{i,k=1}^n \frac{(\lambda_k - \lambda'_i -\imath\varepsilon)}
{S_2\left(\imath(\lambda'_i - \lambda_k +\imath\varepsilon)+\hat{g}|\hat{\bo}\right)
S_2\left(\imath(\lambda_k - \lambda'_i -\imath\varepsilon)|\hat{\bo}\right)}\,
\lim _{x \to +\infty}\,
\lim _{\varepsilon \to 0+}\,
\frac{e^{2\pi \imath
(\underline{\bm{\lambda}}_n-\underline{\bm{\lambda}}'_n)x}}
{\prod_{i,k=1}^n (\lambda_k - \lambda'_i -\imath\varepsilon)} = \\[6pt]
&\lim _{\varepsilon \to 0+}\,
\prod_{i,k=1}^n \frac{(\lambda_k - \lambda'_i -\imath\varepsilon)}
{S_2\left(\imath(\lambda'_i - \lambda_k +\imath\varepsilon)+\hat{g}|\hat{\bo}\right)
S_2\left(\imath(\lambda_k - \lambda'_i -\imath\varepsilon)|\hat{\bo}\right)}\,
\frac{(-1)^{\frac{n(n-1)}{2}}(2\pi \imath )^{n} n!}
{\prod_{k<i}^{n} \left(\lambda_k - \lambda_i\right)^2}\,
\delta\big(\bm{\lambda}_n,\bm{\lambda}'_n \big).
\end{align*}
In the second line we used the fact that function $z^{-1}\,S_2(z|\hat{\bo})$
is regular at the point $z=0$ and extracted the singular part arising
at coinciding arguments $\lambda_k = \lambda'_i$. In the last line we
used Lemma \ref{lemma:delta}.

It remains to calculate the first limit $\ve \rightarrow 0$ and due to the symmetry it can be done at the point
$\bm{\l}'_n = \bm{\l}_n$, i.e. $\lambda'_k =\lambda_k$, $k=1,\ldots,n$.
We have
\begin{align}
&\lim _{\varepsilon \to 0}\,
\prod_{i,k=1}^n \frac{(\lambda_k - \lambda_i -\imath\varepsilon)}
{S_2\left(\imath(\lambda_i - \lambda_k +\imath\varepsilon)+\hat{g}|\hat{\bo}\right)
S_2\left(\imath(\lambda_k - \lambda_i -\imath\varepsilon)|\hat{\bo}\right)} = \\[6pt]
&\Bigl[\frac{\sqrt{\hat{\omega}_1\hat{\omega}_2}}{2\pi \imath}\Bigr]^n\,
\frac{1}{S^n_2\left(\hat{g}|\hat{\bo}\right)}\,
\prod_{i\neq k}^n (\lambda_k - \lambda_i)\,
\prod_{i\neq k}^n \frac{1}
{S_2\left(\imath(\lambda_i - \lambda_k)+\hat{g}|\hat{\bo}\right)
S_2\left(\imath(\lambda_k - \lambda_i)|\hat{\bo}\right)}
\end{align}
where we used formula \eqref{trig5}
\begin{align}
\lim_{z\to 0} z \,S^{-1}_2(z|\hat{\bo}) = \frac{\sqrt{\hat{\omega}_1\hat{\omega}_2}}{2\pi}.
\end{align}
Collecting everything together we obtain
\begin{multline*}
\lim_{x \rightarrow +\infty} \, \lim_{\epsilon \rightarrow 0+}
\left( \Psi_{\bm{\lambda}'_n} | \Psi_{\bm{\lambda}_n}
\right)_{\mu}^{x, \epsilon} = \frac{1}{d_n(g|\bo)}\,
\Bigl[\frac{\sqrt{\hat{\omega}_1\hat{\omega}_2}}{2\pi \imath}\Bigr]^n\,
\frac{1}{S^n_2\left(\hat{g}|\hat{\bo}\right)}\,\times \\[6pt]
\prod_{i\neq k}^n (\lambda_k - \lambda_i)\,
\prod_{i\neq k}^n \frac{1}
{S_2\left(\imath(\lambda_i - \lambda_k)+\hat{g}|\hat{\bo}\right)
S_2\left(\imath(\lambda_k - \lambda_i)|\hat{\bo}\right)}\,
\frac{(-1)^{\frac{n(n-1)}{2}}(2\pi \imath )^{n} n!}
{\prod_{k<i}^{n} \left(\lambda_k - \lambda_i\right)^2}\,
\delta\big(\bm{\lambda}_n,\bm{\lambda}_n \big)\,.
\end{multline*}
Note that
\begin{align*}
\prod_{i\neq k}^n (\lambda_k - \lambda_i) = (-1)^{\frac{n(n-1)}{2}}\,
\prod_{k<i}^{n} \left(\lambda_k - \lambda_i\right)^2
\end{align*}
so that all these factors are cancelled and after some rearrangement we obtain the stated result
\begin{align*}
\lim_{x \rightarrow +\infty} \, \lim_{\epsilon \rightarrow 0+}
\left( \Psi_{\bm{\lambda}'_n} | \Psi_{\bm{\lambda}_n}
\right)_{\mu}^{x, \epsilon} =
\frac{n!\,\delta\big(\bm{\lambda}_n,\bm{\lambda}'_n \big)}
{\prod\limits_{i\neq k}^n S_2\left(\imath(\lambda_i - \lambda_k)+\hat{g}|\hat{\bo}\right)
S_2\left(\imath(\lambda_k - \lambda_i)|\hat{\bo}\right)} = \frac{1}{\hat{\mu}(\bm{\lambda}_{n})}\,
\delta (\bm{\lambda}_{n}, \bm{\lambda}'_{n}).
\end{align*}
This completes the proof of Proposition \ref{proposition-delta}.  \hfill{$\Box$}

Let us return to the complete relation \eqref{T} assuming real periods and coupling constant
\begin{equation*}
\| T\varphi \|_{\mu}^2 = \lim_{x \rightarrow +\infty} \, \lim_{\epsilon \rightarrow 0+}
\langle T\varphi |T\varphi \rangle_{\mu}^{x,\epsilon} = \| \varphi \|_{\hat{\mu}}^2\,.
\end{equation*}
We have proven the second equality and now we prove
the first one
$$
\| T\varphi \|_{\mu}^2 = \lim_{x \rightarrow +\infty} \, \lim_{\epsilon \rightarrow 0+}
\langle T\varphi |T\varphi \rangle_{\mu}^{x,\epsilon}.
$$
This identity means that it is possible to interchange the order of integration and
calculation of the limits.
We divide the proof into two steps.
Firstly we use Fatou's lemma to prove inequality
\begin{equation}
\| T\varphi \|_{\mu}^2  \leq \lim_{x \rightarrow +\infty} \, \lim_{\epsilon \rightarrow 0+}
\langle T\varphi |T\varphi \rangle_{\mu}^{x,\epsilon}.
\end{equation}
Indeed, by Fatou's lemma
\begin{align*}
\langle T\varphi | T\varphi\rangle_{\mu} =&
\int_{\R^n} d\bm{x}_n \, \mu(\bm{x}_n)\,
\varliminf_{x \rightarrow +\infty} \,
\varliminf_{\epsilon \rightarrow 0+}\,e^{\pi \hat{g} (n x - \bbx_n) + 2\pi \epsilon (\bbx_n-nx)} \,
K(x, \bm{x}_n) \, |T\varphi(\bm{x}_n)|^2 \leq \\[8pt]
&\varliminf_{x \rightarrow +\infty} \,
\varliminf_{\epsilon \rightarrow 0+}\,\int_{\R^n} d\bm{x}_n \, \mu(\bm{x}_n)\,
\,e^{\pi \hat{g} (n x - \bbx_n) + 2\pi \epsilon (\bbx_n-nx)} \,
K(x, \bm{x}_n) \, |T\varphi(\bm{x}_n)|^2 =\\[8pt]
& \lim_{x \rightarrow +\infty} \, \lim_{\epsilon \rightarrow 0+}
\langle T\varphi |T\varphi \rangle_{\mu}^{x,\epsilon}
\end{align*}
Both limits inferior equal to usual limits since the latter exist.
Then due to the inequality
$$\langle T\varphi | T\varphi\rangle_{\mu} \leq \| \varphi \|_{\hat{\mu}}^2$$
the integral $\langle T\varphi | T\varphi\rangle_{\mu}$ is absolutely convergent. From \eqref{C} we also have the bound
\begin{align*}
\left|\prod_{k=1}^n e^{\pi \hat{g}(x - x_k) + 2\pi \epsilon (x_k-x)} \,
K(x - x_k)\right| \leq C^n\,.
\end{align*}
Hence, we can use dominated convergence theorem to interchange the integral with the limits
\begin{align*}
\langle T\varphi | T\varphi\rangle_{\mu} =&
\int_{\R^n} d\bm{x}_n \, \mu(\bm{x}_n)\,
\lim_{x \rightarrow +\infty} \,
\lim_{\epsilon \rightarrow 0+}\,e^{\pi \hat{g} (n x - \bbx_n) + 2\pi \epsilon (\bbx_n-nx)} \,
K(x, \bm{x}_n) \, |T\varphi(\bm{x}_n)|^2 = \\[8pt]
&\lim_{x \rightarrow +\infty} \,
\lim_{\epsilon \rightarrow 0+}\,\int_{\R^n} d\bm{x}_n \, \mu(\bm{x}_n)\,
\,e^{\pi \hat{g} (n x - \bbx_n) + 2\pi \epsilon (\bbx_n-nx)} \,
K(x, \bm{x}_n) \, |T\varphi(\bm{x}_n)|^2 =\\[8pt]
& \lim_{x \rightarrow +\infty} \, \lim_{\epsilon \rightarrow 0+}
\langle T\varphi |T\varphi \rangle_{\mu}^{x,\epsilon}\,.
\end{align*}
This completes the proof of Theorem \ref{T-theorem}.  \hfill{$\Box$}

Consider now a linear map
\begin{align}
(S \phi) (\bm{\lambda}_n) = \int d\bm{x}_n \, \mu(\bm{x}_n) \,
\Psi_{\bm{\lambda}_n} (\bm{x}_n) \, \phi(\bm{x}_n)
\end{align}
originally defined on smooth functions with compact support. It is remarkable that duality~\eqref{duality}
\begin{equation*}
\Psi_{\bm{\lambda}_n}(\bm{x}_n; g|\bo) = \Psi_{\bm{x}_n}(\bm{\lambda}_n; \gg|\hat{\bo})
\end{equation*}
states that formulae for $S$-transform are obtained from the formulae for $T$-transform
by change of variables
$\bx_{n} \rightleftarrows \bm{\lambda}_n$ and parameters
$g\,,\bo \rightleftarrows \gg\,,\hat{\bo}$ and for that change we have
$\hat{\mu}(\bm{\lambda}_n) \rightleftarrows \mu(\bm{x}_n)$.

\Stheorem*

Here recall that $R$ is the total reflection operator in $\R^n$ \eqref{R}
\begin{equation}
R \vf(\bl_n) = \vf(-\bl_n).
\end{equation}
The proof of the first part of Theorem \ref{S-theorem} and of the second orthogonality relation
\begin{align}
&\int_{\R^{n}} d\bl_{n}\,
\hat{\mu}(\bl_{n})\,
\overline{\Psi}_{\bm{\lambda}_n}(\bm{x}_n)
\Psi_{\bm{\lambda}_{n}}(\bm{y}_{n}) =
\mu^{-1}(\bm{x}_{n})\,\delta (\bm{x}_{n}, \bm{y}_{n})
\end{align}
repeats the proof of Theorem \ref{T-theorem}.

The second part of the theorem is an immediate corollary of the relation
\beq\label{dseq8}
 \overline{\Psi}_{\bm{\lambda}_n} (\bm{x}_n) =
\Psi_{-\bm{\lambda}_n} (\bm{x}_n)=\Psi_{\bm{\lambda}_n} (-\bm{x}_n).
\eeq
Indeed, this relation implies that the  operator $RS$ coinsides with the adjoint to $T$
\beq \label{ds11} RS=T^*, \eeq
so that the adjoint operator $T^*$ is defined on the whole space
$L^2_{\text{sym}}\left(\mathbb R^{n}, \mu\right)$. Moreover, using Theorem \ref{T-theorem}  and  the relation \rf{phipsi}  we get
\begin{equation}
\langle T^* T \vf | \psi \rangle_{\hat{\mu}} = \langle T \vf | T \psi \rangle_\mu = \langle \vf | \psi \rangle_{\hat{\mu}},
\end{equation}
that is $\vf - T^* T \vf$ is orthogonal to everything and therefore equals zero. Hence,
$$ T^\ast T= \mathrm{Id}.$$ \cbk
The last relation is equivalent to
$$RST=STR=\mathrm{Id}.$$
Applying the same  argument to the operator $S^\ast$, we get the opposite relation
$$RTS=TSR=\mathrm{Id}$$
and the proof of the second part of the theorem.  \hfill{$\Box$}

\section*{Acknowledgments}
The work of N. Belousov and S. Derkachov was supported by Russian Science Foundation, project No. 23-11-00311, used for the proof of statements of Section 3. The
work of S. Khoroshkin was supported by Russian Science Foundation,
project No 23-11-00150, used for the proof of statements of Section 2.
The work of S. Kharchev was supported by Russian Science Foundation,
project No 20-12-00195, used for the
proof of statements of Section 1 and Appendices A, B, C. The authors also thank
Euler International Mathematical Institute for hospitality during the PDMI and HSE
joint workshop on quantum integrability, where they got a possibility to discuss many
subtle points of this work.
\cbk

\setcounter{equation}{0}

\section*{Appendix}
\appendix

\setcounter{equation}{0}
\section{The double sine function}\label{S2}
The  double sine  function $S_2(z):=S_2(z|\bo)$, see \cite{Ku} and references therein, is a meromorphic function that satisfies two functional relations
\beq\label{trig3}  \frac{S_2(z)}{S_2(z+\o_1)}=2\sin \frac{\pi z}{\o_2},\qquad \frac{S_2(z)}{S_2(z+\o_2)}=2\sin \frac{\pi z}{\o_1}
\eeq
and inversion relation
\beq \label{trig4} S_2(z)S_2(-z)=-4\sin\frac{\pi z}{\o_1}\sin\frac{\pi z}{\o_2},\eeq
or equivalently
\beq \label{inv} S_2(z)S_2(\o_1+\o_2-z)=1. \eeq
The function $S_2(z)$ has poles at the points
\beq \label{S-poles}
z_{m,k} = m \o_1 + k\o_2, \qquad m,k\geq 1
\eeq
and zeros at
\beq\label{S-zeros}
z_{-m,-k}=-m\o_1-k\o_2,\qquad m,k\geq 0.
\eeq
For $\o_1 / \o_2 \not\in \mathbb{Q}$ all poles and zeros are simple. The residues of $S_2(z)$ and $S^{-1}_2(z)$ at these points are
\begin{align}
\underset{z = z_{m,k}}{\Res} \, S_2(z) = \frac{\sqrt{\o_1\o_2}}{2\pi}\frac{(-1)^{mk}}{\prod\limits_{s=1}^{m - 1}2\sin\dfrac{\pi s\o_1}{\o_2}\prod\limits_{l=1}^{k - 1}2\sin\dfrac{\pi l\o_2}{\o_1}},
\\[10pt]
\label{trig5} \underset{z = z_{-m,-k}}{\Res} \, S^{-1}_2(z) = \frac{\sqrt{\o_1\o_2}}{2\pi}\frac{(-1)^{mk+m+k}}{\prod\limits_{s=1}^m2\sin\dfrac{\pi s\o_1}{\o_2}\prod\limits_{l=1}^k2\sin\dfrac{\pi l\o_2}{\o_1}}.
\end{align}
In the analytic region $ \Re z \in ( 0, \omega_1 + \omega_2 )$ we have the following integral representation for the logarithm of $S_2(z)$
\begin{equation}\label{S2-int}
\ln S_2 (z) = \int_0^\infty \frac{dt}{2t} \left( \frac{\sh \left[ (2z - \omega_1 - \omega_2)t \right]}{ \sh (\omega_1 t) \sh (\omega_2 t) } - \frac{ 2z - \omega_1 - \omega_2 }{ \omega_1 \omega_ 2 t } \right).
\end{equation}
It is clear from this representation that the double sine function is homogeneous
\beq\label{S-hom}
S_2( \gamma z | \gamma\o_1, \gamma \o_2 ) = S_2(z|\o_1, \o_2), \qquad \gamma \in (0, \infty)
\eeq
and invariant under permutation of periods
\beq\label{A6}
S_2(z| \o_1, \o_2) = S_2(z | \o_2, \o_1).
\eeq
The double sine function can be expressed through the Barnes double Gamma function $\Gamma_2(z|\bo)$ \cite{B},
\beq
S_2(z|\bo)=\Gamma_2(\o_1+\o_2-z|\bo)\Gamma_2^{-1}(z|\bo),
\eeq
and its properties follow from the corresponding properties of the double Gamma function.

It is also connected to the Ruijsenaars hyperbolic Gamma function $G(z|\bo)$ \cite{R1}
\beq \label{G-S}
G(z|\bo) = S_2\Bigl(\imath z + \frac{\o_1 + \o_2}{2} \,\Big|\, \bo \Bigr)
\eeq
and to the Faddeev quantum dilogarithm $\gamma(z|\bo)$ \cite{F}
\beq \notag
\gamma(z|\bo) = S_2\Bigl(-\imath z + \frac{\o_1+\o_2}{2}\, \Big|\, \bo\Bigr) \exp \Bigl( \frac{\imath \pi}{2\o_1 \o_2} \Bigl[z^2 + \frac{\o_1^2+\o_2^2}{12} \Bigr]\Bigr).
\eeq
Both $G(z|\bo)$ and $\gamma(z|\bo)$ were investigated independently.

In the paper we deal only with ratios of double sine functions denoted by $\mu(x)$ \eqref{I5} and $K(x)$ \eqref{I6}
\beq\label{B1}
\begin{split}\mu(x)& =S_2(\imath x)S_2^{-1} (\imath x+g),\\[6pt]
	K(x)& =  S_2\left(\imath x+\frac{\o_1+\o_2}{2}+\frac{g}{2}\right)S_2^{-1}\left(\imath x+\frac{\o_1+\o_2}{2}-\frac{g}{2}\right).
\end{split}
\eeq
Now we will give the key asymptotic formulas and bounds for them, which were derived in \cite[Appendices A, B]{BDKK} from the known results for the double Gamma function. In what follows we assume conditions
\begin{equation}\label{cond}
\Re \o_j > 0, \qquad 0 < \Re g < \Re\o_1 + \Re \o_2, \qquad \Re \hat{g} > 0,
\end{equation}
where we denoted
\begin{equation}
\hat{g} = \frac{g}{\o_1 \o_2}.
\end{equation}
Let $\sigma_i$ be the arguments of the periods~$\o_i$, $|\sigma_i|<\pi/2$. Because of the symmetry \eqref{A6}, we may assume that $\sigma_1 \geq \sigma_2$. Let $D_+$ and $D_-$ be the cones of poles \eqref{S-poles} and zeros \eqref{S-zeros} of~$S_2(z)$:
\beqq D_+=\{ z\colon \sigma_2< \arg z<\sigma_1\},\qquad D_-=\{ z\colon \pi+\sigma_2< \arg z<\pi+\sigma_1\},\qquad D=D_+\cup D_-.\eeqq
Denote by $d(z,D)$ the distance between a point $z$ and the cones $D$. Using Barnes' Stirling formula for the asymptotic of double Gamma function, for the ratio of double sines one obtains
\beq \label{A14} \frac{S_2(z)}{S_2(z+g)} =  e^{\mp\pi\imath \hat{g}\left(z-\frac{g^\ast}{2}\right)}\Bigl(1+\,O\Bigl(d^{-1}(z,D)\Bigr)\Bigr), \eeq
where $z \in \mathbb{C} \setminus D$ and the sign $-$ (or $+$) is taken for $\Re z > 0$ (or $\Re z < 0$), see \cite[eq.(A.18)]{BDKK}.
Then from \eqref{A14} for the functions $\mu(x)$ and $K(x)$ \eqref{B1} with $x \in \R$ we have
\begin{equation}\label{Kmu-asymp}
\mu(x) \sim e^{\pi \hat{g} | x | \pm \imath \frac{\pi \hat{g} g^*}{2} }, \qquad K(x) \sim e^{- \pi \hat{g} |x|}, \qquad x\rightarrow \pm \infty.
\end{equation}
Denote also
\beq \nu_g= \Re  \hat{g}.\eeq
Under the conditions \eqref{cond} by using the same Stirling formula we also have bounds
\beq \label{Kmu-bound}
|\mu(x)| \leq C e^{\pi\nu_g |x|}, \qquad |K(x)| \leq C e^{-\pi\nu_g |x|},  \qquad x \in \R
\eeq
where $C$ is a positive constant uniform for compact subsets of parameters $\bo, g$ preserving the mentioned conditions, see \cite[eq.(B.3)]{BDKK}.

Another key result that we need in the paper is the following Fourier transform formula given in \cite[Proposition C.1]{R3}, which we rewrite in terms of the double sine function using connection formula \eqref{G-S}.  This Fourier transform can be already found in \cite{FKV,PT}.
\begin{proposition*}\cite{R3}
For real positive periods $\o_1, \o_2$ we have
\begin{equation}
\begin{aligned}
&\int_{\mathbb{R}} dx \, e^{\frac{2\pi\imath}{\o_1 \o_2} y x} S_2\Bigl(\imath x - \imath \nu + \frac{\o_1 + \o_2}{2} \Bigr) S_2^{-1} \Bigl( \imath x - \imath \rho + \frac{\o_1 + \o_2}{2} \Bigr) \\[6pt]
&= \sqrt{\o_1 \o_2} \, e^{ \frac{\pi\imath}{\o_1 \o_2} y (\nu + \rho) } S_2(\imath \rho - \imath \nu) \, S_2^{-1}\Bigl(\imath y + \frac{\imath(\rho - \nu)}{2} \Bigr) \, S_2^{-1}\Bigl(-\imath y + \frac{\imath(\rho - \nu)}{2} \Bigr),
\end{aligned}
\end{equation}
while the parameters $\nu, \rho, y$ satisfy the conditions
\begin{equation}\label{A19}
-\frac{\o_1 + \o_2}{2} < \Im \rho < \Im \nu < \frac{\o_1 + \o_2}{2}, \qquad |\Im y | < \Im \frac{\nu - \rho}{2}.
\end{equation}
\end{proposition*}
In the special case
\begin{equation}
\nu = \frac{\imath g}{2}, \qquad \rho = -\frac{\imath g}{2}
\end{equation}
taking $y = \o_1 \o_2 \l$ and using homogeneity of the double sine \eqref{S-hom} (with $\gamma = \o_1 \o_2$) we arrive at the Fourier transform formula for the function $K(x)$ \eqref{B1}
\begin{equation}\label{K-fourier}
\int_{\mathbb{R}} dx \; e^{2 \pi \imath  \lambda x }  \K (x) = \sqrt{\omega_1 \omega_2} \, S_2(g) \, \KK (\lambda),
\end{equation}
where $| \Im \l | < \Re \hat{g}/2$ and conditions \eqref{A19} are satisfied due to the inequalities on the coupling constant $g$ \eqref{cond}. Here we recall the notations
\begin{equation}
\hat{K}(\l) = K_{\hat{g}^*}(\l|\hat{\o}), \qquad \hat{g}^* = \frac{g^*}{\o_1\o_2}, \qquad \hat{\bo} = \Bigl( \frac{1}{\o_2}, \frac{1}{\o_1} \Bigr).
\end{equation}
Note that the right hand side of \eqref{K-fourier} is analytic function of $\o_1, \o_2$ in the domain $\Re \o_j > 0$. The integral from the left is also analytic with respect to periods. Indeed, due to the bound \eqref{Kmu-bound} it is absolutely convergent uniformly on compact sets of parameters $\bo, g$ preserving the conditions \eqref{cond}. Hence, the formula \eqref{K-fourier} also holds for complex periods under the mentioned conditions.
\cbk

\setcounter{equation}{0}
\section{Some inequalities}\label{AppB}
Let
\begin{equation}\label{Ln}
L_n(\bm{y}_{n}, \bm{x}_{n + 1}) = \sum_{\substack{i, j = 1 \\ i < j}}^{n + 1} | x_i - x_j | + \sum_{\substack{i, j = 1 \\ i < j}}^{n} | y_i - y_j | - \sum_{i = 1}^{n+1} \sum_{j = 1}^{n} \, | x_i - y_j |.
\end{equation}
\begin{lemma} The following inequality holds
\begin{equation}
L_n \leq 0
\end{equation}
for any $x_j, y_j \in \R$.
\end{lemma}
We found the short proof of the this inequality in \cite[Appendix A]{IS}, and we repeat it here for the self-consistency of the paper.
\begin{proof}
The function $L_n$ is clearly symmetric with respect to $x_j$ and $y_j$ (separately). Then without loss of generality consider the case
\begin{equation}
x_1 \geq x_2 \geq \ldots \geq x_{n+1}, \qquad y_1 \geq y_2 \geq \ldots \geq y_{n}.
\end{equation}
By simple combinatorics in this case we have
\begin{equation}
\sum_{\substack{i, j = 1 \\ i < j}}^{n + 1} | x_i - x_j | = \sum_{k=1}^{n + 1} (n + 2 - 2k) x_k, \qquad \sum_{\substack{i, j = 1 \\ i < j}}^{n} | y_i - y_j | = \sum_{k=1}^{n} (n + 1 - 2k) y_k.
\end{equation}
Therefore, we can write
\begin{equation}
\sum_{\substack{i, j = 1 \\ i < j}}^{n + 1} | x_i - x_j | + \sum_{\substack{i, j = 1 \\ i < j}}^{n} | y_i - y_j | = \sum_{k = 1}^{n + 1} \Biggl( \sum_{m = 1}^{k - 1}(y_m - x_k) + \sum_{m = k}^n (x_k - y_m) \Biggr).
\end{equation}
Inserting this in the function \eqref{Ln},
\begin{equation}
L_n = \sum_{k = 1}^{n + 1} \Biggl( \sum_{m = 1}^{k - 1}\bigl[ (y_m - x_k) - |y_m - x_k | \bigl]+ \sum_{m = k}^n \bigl[ (x_k - y_m) - |x_k - y_m| \bigr] \Biggr).
\end{equation}
Then $L_n \leq 0$, since $a \leq |a|$.
\end{proof}

Now consider
\begin{equation}
R_n(\by_n, \bx_n) = \sum_{\substack{i, j = 1 \\ i < j}}^{n} | x_i - x_j | + \sum_{\substack{i, j = 1 \\ i < j}}^{n} | y_i - y_j | - \sum_{i,j = 1}^{n} \, | x_i - y_j |
\end{equation}
and recall the notation
\begin{equation}
\| \bx_n \| = |x_1| + \ldots + |x_n|.
\end{equation}
\begin{corollary}\label{cor:R-ineq}
For any $\epsilon \in [0,1]$
\begin{equation}
R_n \leq \epsilon ( \|\bm{y}_n \| - \| \bm{x}_n \|).
\end{equation}
\end{corollary}
\begin{proof}
Consider $L_{n}(\by_n, \bx_{n+1})$ with $x_{n + 1} = 0$
\begin{equation}
\begin{aligned}
L_{n}(\bm{y}_n, \bm{x}_n, 0) = \sum_{\substack{i, j = 1 \\ i < j}}^{n} | x_i - x_j | + \sum_{\substack{i, j = 1 \\ i < j}}^{n} | y_i - y_j | &- \sum_{i,j = 1}^{n} \, | x_i - y_j |\\[6pt]
& + \|\bm{x}_n\| - \| \bm{y}_n \| \leq 0.
\end{aligned}
\end{equation}
Consequently,
\begin{equation}
R_n \leq \| \bm{y}_n \| - \| \bm{x}_n \|
\end{equation}
and by symmetry
\begin{equation}
R_n \leq \| \bm{x}_n \| - \| \bm{y}_n \|.
\end{equation}
Summing these inequalities we also have
\begin{equation}
R_n \leq 0.
\end{equation}	
Therefore, for any $\eps \in [0,1]$
\begin{equation}
R_n = (1 - \epsilon) R_n + \epsilon R_n \leq \epsilon ( \|\bm{y}_n \| - \| \bm{x}_n \|).
\end{equation}
\end{proof}

\setcounter{equation}{0}
\section{Delta sequence}\label{App-delta}

We are going to show that in the sense of distributions the following identity holds
\begin{align}\label{id}
\lim_{\lambda\rightarrow\infty} \, \lim_{\epsilon\rightarrow 0+}
\frac{e^{\imath \lambda \sum_{a=1}^{n}(x_a-y_a)}}{
\prod_{a,b=1}^{n} \left(x_a-y_b -\imath\varepsilon\right)}
= \frac{(-1)^{\frac{n(n-1)}{2}}(2\pi \imath)^{n} n!}{ \prod_{a<b}^{n} \left(x_{a}-x_{b}\right)^2}\,
\delta\big(\bm{x}_n,\bm{y}_n \big),
\end{align}
where
\begin{align*}
\delta\big(\bm{x}_n,\bm{y}_n \big) = \frac1{n!}\,
\sum_{w\in S_n} \prod_{k=1}^n\delta(x_k-y_{w(k)}).
\end{align*}
This identity is written in a compact formal way and should be understood
in the following sense:
for any test (e.g. smooth with a compact support) function $f(x_1,\ldots,x_n)$ we have
\begin{multline}\label{lem}
\lim_{\lambda\rightarrow\infty} \, \lim_{\epsilon\rightarrow 0+}
\int dx_1\cdots dx_n \prod_{a<b}^{n} (x_{a}-x_{b})^2\,f(x_1,\ldots,x_n)\,
\frac{e^{\imath \lambda \sum_{a=1}^{n}(x_a-y_a)}}{
\prod_{a,b=1}^{n} \left(x_a-y_b -\imath\varepsilon\right)}
= \\ (-1)^{\frac{n(n-1)}{2}}(2\pi \imath)^{n}\,
\sum_{w\in S_n} f\big(y_{w(1)},\ldots,y_{w(n)}\big).
\end{multline}
First of all we prove the equivalent identity
\begin{align}\label{0}
\lim_{\lambda\rightarrow\infty}\lim_{\epsilon\rightarrow 0^+}
e^{\imath \lambda \sum_{a=1}^{n}(x_a-y_a)}\,\frac{
\prod_{a<b}^{n} \left(x_{a}-x_{b}\right)\left(y_{b}-y_{a}\right)}{
\prod_{a,b=1}^{n} \left(x_a-y_b -\imath\varepsilon\right)}
= (2\pi \imath )^{n}\,
\sum_{w\in S_n} (-1)^{s(w)}\prod_{k=1}^n\,\delta(x_k-y_{w(k)})
\end{align}
where $s(w)$ is the sign of the permutation $w$.
Let us start from the simplest example $n=1$. We have to prove that
\begin{align*}
\lim_{\lambda\rightarrow\infty}\lim_{\epsilon\rightarrow 0^+}
\frac{e^{\imath \lambda (x-y)}}{\left(x-y -\imath\varepsilon\right)}
= 2\pi \imath\,\delta(x - y )\,,
\end{align*}
or equivalently
\begin{align*}
\lim_{\lambda\rightarrow\infty}\lim_{\epsilon\rightarrow 0^+}
\int_{\mathbb{R}} f(x)\,
\frac{e^{\imath \lambda (x-y)}}{x-y -\imath\varepsilon}\,dx = 2\pi \imath\,f(y)\,.
\end{align*}
First of all we transform integral with the test function.
We divide integral on two parts: the first
integral can be calculated by residues and due
to cancelation of singularity at $x=y$ it is possible to put
$\varepsilon \to 0 $ in the second part
\begin{multline*}
\int_{\mathbb{R}} f(x)\,
\frac{e^{\imath\lambda(x-y)}}{x-y -\imath\varepsilon}\,dx =
f(y)\,\int_{\mathbb{R}}\,
\frac{e^{\imath\lambda(x-y)}}{x-y-\imath\varepsilon}\,dx  +
\int_{\mathbb{R}} \frac{f(x)-f(y)}
{x-y-\imath\varepsilon}\,e^{\imath\lambda(x-y)}\,dx = \\[6pt]
2\pi \imath\,f(y)\,e^{-\varepsilon\lambda}  +
\int_{\mathbb{R}} \frac{f(x)-f(y)}
{x-y-\imath\varepsilon}\,e^{\imath\lambda(x-y)}\,dx \xrightarrow{\varepsilon\to 0}
2\pi \imath\,f(y)  +
\int_{\mathbb{R}} \frac{f(x)-f(y)}
{x-y}\,e^{\imath \l (x- y)}\,dx.
\end{multline*}
Due to the Riemann-Lebesgue lemma the second contribution
tends to zero in the limit $\l\to\infty$ so that we obtain
after removing $\varepsilon$-regularization and taking limit $\lambda\to\infty$
\begin{align*}
\lim_{\lambda\rightarrow\infty} \, \lim_{\epsilon\rightarrow 0+}
\int_{\mathbb{R}} f(x)\,
\frac{e^{\imath\lambda(x-y)}}{x-y -\imath\varepsilon}\,dx = 2\pi \imath\,f(y)\,.
\end{align*}
The whole consideration in the case $n=2$ is almost identical
to the case of general $n$.
We have to prove the following relation
\begin{align*}
\lim_{\lambda\rightarrow\infty} \, \lim_{\epsilon\rightarrow 0+}
\frac{e^{\imath\lambda(x_1+x_2-y_1-y_2)}\,x_{12}\,y_{21}}{
\prod_{a,b=1}^{2} \left(x_a-y_b -\imath\varepsilon\right)}
= (2\pi \imath)^{2}\,
\bigl[\delta(x_1-y_1)\,\delta(x_2-y_2)-\delta(x_1-y_2)\,\delta(x_2-y_1)
\bigl].
\end{align*}
Here we denoted $x_{12} = x_1 - x_2$. First of all we use Cauchy determinant identity
\begin{multline}\label{12}
\frac{x_{12}\,y_{21}}{
\prod_{a,b=1}^{2} \left(x_a-y_b -\imath\varepsilon\right)} =
\frac{1}{\left(x_1-y_1 -\imath\varepsilon\right)
\left(x_2-y_2 - \imath\varepsilon\right)} -
\frac{1}{\left(x_1-y_2 -\imath\varepsilon\right)
\left(x_2-y_1 -\imath\varepsilon\right)}.
\end{multline}
Let us consider the convolution of the first term with
the test function
\begin{align*}
\int dx_1\,dx_2\,f(x_1,x_2)\,
\frac{e^{\imath\lambda(x_1+x_2-y_1-y_2)}}{\left(x_1-y_1 -\imath\varepsilon\right)
\left(x_2-y_2 -\imath\varepsilon\right)}
\end{align*}
and introduce two commuting operators $X_1$ and $X_2$ acting on the test function
\begin{align*}
X_1 f(x_1\,,x_2) = f(y_1\,,x_2), \qquad
X_2 f(x_1\,,x_2) = f(x_1\,,y_2)\,.
\end{align*}
As a consequence of evident identity
\begin{multline*}
1 = (1-X_1+X_1)(1-X_2+X_2) =
(1-X_1)(1-X_2) + X_1(1-X_2) + X_2(1-X_1) + X_1 X_2
\end{multline*}
and explicit formulas
\begin{align*}
& X_1(1-X_2)f(x_1\,,x_2) = f(y_1\,,x_2) - f(y_1\,,y_2)\,;\\[6pt]
& X_2(1-X_1)f(x_1\,,x_2) = f(x_1\,,y_2) - f(y_1\,,y_2)\,;\\[6pt]
& (1-X_1)(1-X_2)f(x_1\,,x_2) =
(1-X_1)\left[f(x_1\,,x_2) - f(x_1\,,y_2)\right] = \\[6pt]
& f(x_1\,,x_2) - f(x_1\,,y_2) - f(y_1\,,x_2) + f(y_1\,,y_2)
\end{align*}
we obtain the following useful representation for the function $f(x_1\,,x_2)$
\begin{align*}
f(x_1\,,x_2) = f(y_1\,,y_2) + \left[f(y_1\,,x_2) - f(y_1\,,y_2)\right] +
\left[f(x_1\,,y_2) - f(y_1\,,y_2)\right] + \\[6pt]
\left[f(x_1\,,x_2) - f(x_1\,,y_2) - f(y_1\,,x_2) + f(y_1\,,y_2)\right]
\end{align*}
Note that the first term does not depend on $x_1$ and $x_2$,
second term does not depend on $x_1$ and is equal to zero at the point $x_2=y_2$,
third term does not depend on $x_2$ and is equal to zero at the point $x_1=y_1$.
The Taylor expansion of the last term in vicinity of the point $x_1=y_1\,,x_2=y_2$ starts from the contribution $\sim(x_1-y_1)(x_2-y_2)$ because it turns to zero at points $x_1=y_1$ and $x_2=y_2$ independently. In the first three terms the corresponding integrals can be calculated by residues
and we obtain
\begin{multline*}
\int dx_1\,dx_2\,
\frac{f(x_1\,,x_2)\,e^{\imath\lambda(x_1+x_2-y_1-y_2)}}{\left(x_1-y_1 -\imath\varepsilon\right)
\left(x_2-y_2 -\imath\varepsilon\right)} = f(y_1\,,y_2)\,(2\pi \imath)^2\,e^{-2\varepsilon\lambda} + \\[6pt]
2\pi \imath\,e^{-\varepsilon\lambda}\,
\int dx_2\,
\frac{\left[f(y_1\,,x_2) - f(y_1\,,y_2)\right]\,e^{\imath\lambda(x_2-y_2)}}{x_2-y_2 -\imath\varepsilon}
+ \\[6pt]
2\pi \imath\,e^{-\varepsilon\lambda}\,
\int dx_1\,
\frac{\left[f(x_1\,,y_2) - f(y_1\,,y_2)\right]\,
e^{\imath\lambda(x_1-y_1)}}{x_1-y_1 -\imath\varepsilon} +
\\[6pt]
\int dx_1\,dx_2\,
\frac{\left[f(x_1\,,x_2) - f(x_1\,,y_2) - f(y_1\,,x_2) + f(y_1\,,y_2)\right]\,e^{\imath\lambda(x_1+x_2-y_1-y_2)}}{\left(x_1-y_1 -\imath\varepsilon\right)
\left(x_2-y_2 -\imath\varepsilon\right)}.
\end{multline*}
Inside of remaining integrals all singularities of integrand
are cancelled so that it is possible to perform the limit $\varepsilon \to 0$.
Due to the Riemann-Lebesgue lemma all contributions with integrals
tend to zero in the limit $\lambda\to\infty$ and we have
after removing $\varepsilon$-regularization and $\lambda\to\infty$
\begin{align*}
\lim_{\lambda\rightarrow\infty}\lim_{\epsilon\rightarrow 0^+}
\int dx_1\,dx_2\,
\frac{f(x_1\,,x_2)\,e^{\imath\lambda(x_1+x_2-y_1-y_2)}}{\left(x_1-y_1 -\imath\varepsilon\right)
\left(x_2-y_2 -\imath\varepsilon\right)} = (2\pi \imath)^2\,f(y_1\,,y_2)
\end{align*}
The second term in \eqref{12} is obtained by $y_1\rightleftarrows y_2$
so that finally one obtains the stated result
\begin{multline*}
\lim_{\lambda\rightarrow\infty}\lim_{\epsilon\rightarrow 0+}
\int dx_1\,dx_2\,
\frac{f(x_1\,,x_2)\,e^{\imath\lambda(x_1+x_2-y_1-y_2)} \, x_{12}y_{21}}
{\left(x_1-y_1 -\imath\varepsilon\right)\left(x_1-y_2 -\imath\varepsilon\right)
\left(x_2-y_1 -\imath\varepsilon\right)
\left(x_2-y_2 -\imath\varepsilon\right)} = \\[6pt]
(2\pi \imath)^2\,\left[ f(y_1\,,y_2)- f(y_2\,,y_1) \right].
\end{multline*}
It is evident that the symmetric part of the function $f(x_1,x_2)$ does not contribute so that
the nontrivial contribution is due to antisymmetric part of the function $f(x_1,x_2)$.
Antisymmetric part of the test function $f(x_1,x_2)$ should be zero at $x_1=x_2$ and
without loss of generality it is possible to use representation
$f(x_1,x_2) = (x_1-x_2)\phi(x_1,x_2)$, where $\phi(x_1,x_2)$ can be generic because
antisymmetric part of the function $\phi(x_1,x_2)$ does not contribute.
Finally one obtains the formula \eqref{lem} in the case $n=2$
\begin{multline*}
\lim_{\lambda\rightarrow\infty}\lim_{\epsilon\rightarrow 0+}
\int dx_1\,dx_2\,
\frac{x^2_{12}\,\phi(x_1\,,x_2)\,e^{\imath\lambda(x_1+x_2-y_1-y_2)}}
{\left(x_1-y_1 -\imath\varepsilon\right)\left(x_1-y_2 -\imath\varepsilon\right)
\left(x_2-y_1 -\imath\varepsilon\right)
\left(x_2-y_2 -\imath\varepsilon\right)} = \\[6pt]
(2\pi \imath)^2\,
\left[ \phi(y_1\,,y_2) + \phi(y_2\,,y_1) \right].
\end{multline*}
In general case we again use Cauchy determinant identity in the form
\begin{align*}
\frac{\prod_{k<j}x_{kj}\,y_{jk}}{\prod_{k,j=1}^{n}(x_k-y_j-\imath\varepsilon)} &
= \det\left(\frac{1}{x_k-y_{j}-\imath\varepsilon}\right) =
\sum_{\sigma\in S_{n}} (-1)^{s(\sigma)}
\prod_{k=1}^{n}\frac{1}{x_k-y_{\sigma(k)}-\imath\varepsilon}
\end{align*}
In analogy with $n=2$ we shall prove that
\begin{align*}
\lim_{\lambda\rightarrow\infty}\lim_{\epsilon\rightarrow 0+}
\int dx_1\,\cdots dx_n\,
\frac{f(x_1,\ldots ,x_n)\,e^{\imath\lambda\sum_k(x_k-y_k)}}{
\prod_k\left(x_k-y_k -\imath\varepsilon\right)} = (2\pi \imath)^n\,f(y_1,\ldots,y_n)
\end{align*}
and then use the same identity with evident permutations.
We introduce the natural generalization of the operators $X_k$
\begin{align*}
X_k f(x_1,\ldots ,x_k,\ldots ,x_n) = f(x_1,\ldots ,y_k,\ldots ,x_n)
\end{align*}
and the main expansion
\begin{multline*}
1 = \prod_{k=1}^N (1-X_k+X_k) = \\
\prod_{k=1}^N (1-X_k) +\sum_{k=1} X_k \prod_{i\neq k}^N (1-X_i)+
\sum_{k,p=1} X_k X_p \prod_{i\neq k,p}^N (1-X_i)+ \ldots + \prod_{k=1}^N X_k
\end{multline*}
Due to the Riemann-Lebesgue lemma all contributions containing
$\prod_{i} (1-X_i)\,f(x_1,\ldots,x_n)$ in integrand are regular
at corresponding points so that corresponding integrals
tend to zero in the limit $\l\to\infty$. In the needed limit only one term
$\prod_{k=1}^n X_k$ survives and produces $(2\pi \imath)^n\,f(y_1,\ldots,y_n)$ in the full anagoly with the case $n=2$.
Then for the whole sum we obtain
\begin{multline}\label{lem0}
\lim_{\lambda\rightarrow\infty}\lim_{\epsilon\rightarrow 0+}
\int dx_1\cdots dx_n\, f(x_1,\ldots,x_n)\,
\frac{\prod_{k<j}x_{kj}\,y_{jk}\,e^{\imath\lambda\sum_{a=1}^{n}(x_a-y_a)}}{
\prod_{a,b=1}^{n} \left(x_a-y_b -\imath\varepsilon\right)}
= \\[6pt]
(-1)^{\frac{n(n-1)}{2}}(2\pi \imath )^{n}\,
\sum_{w\in S_n} (-1)^{s(w)}\,
f\big(y_{w(1)},\ldots,y_{w(n)}\big)
\end{multline}
and this identity is equivalent to \eqref{0}.
The next step is very similar to the case $n=2$.
Indeed, only the antisymmetric part of the test function
$f(x_1,\ldots,x_n)$ gives nontrivial contribution so that without
loss of generality it is possible to use the following
representation for the test function
$f(x_1,\ldots,x_n) = \Delta(x_1,\ldots,x_n)\,\phi(x_1,\ldots,x_n)$,
where $\Delta(x_1,\ldots,x_n) = \prod_{k<j} x_{kj}$.
We have an evident relation
\begin{align*}
\Delta\big(x_{w(1)},\ldots,x_{w(n)}\big) = (-1)^{s(w)}\,
\Delta\big(x_{1},\ldots,x_{n}\big)
\end{align*}
and as consequence one obtains \eqref{lem}
\begin{multline*}
\lim_{\lambda\rightarrow\infty}\lim_{\epsilon\rightarrow 0+}
\int dx_1\cdots dx_n\, \phi(x_1,\ldots,x_n)\,
\frac{\prod_{k<j}x^2_{kj}\,e^{\imath\lambda\sum_{a=1}^{n}(x_a-y_a)}}{
\prod_{a,b=1}^{n} \left(x_a-y_b -\imath\varepsilon\right)}
= \\[6pt]
(-1)^{\frac{n(n-1)}{2}}(2\pi \imath )^{n}\,
\sum_{w\in S_n} \,
\phi\big(y_{w(1)},\ldots,y_{w(n)}\big).
\end{multline*}

\end{document}